\documentclass{article}

\usepackage{amsmath,amsthm} 
\usepackage{amssymb,mathrsfs} 
\usepackage{graphicx} 
\usepackage{color,subfigure} 
\usepackage{enumerate,mathtools}  
\usepackage{fullpage}

\newtheorem{theorem}{Theorem}
\newtheorem{lemma}{Lemma}
\newtheorem{assumption}{Assumption}
\newtheorem{prop}{Proposition}

\newtheorem{remark}[theorem]{Remark}

\DeclareMathAlphabet{\mathpzc}{OT1}{pzc}{m}{it}
 
\newcommand{\rme}{\mathrm{e}} 
\newcommand{\ri}{\mathrm{i}} 
\newcommand{\eps}{\epsilon}

\newcommand{\bZ}{\mathbb{Z}}
\newcommand{\bN}{\mathbb{N}}

\newcommand{\RR}{\mathbb{R}}

\newcommand{\cA}{\mathcal{A}}

\newcommand{\cH}{\mathcal{H}}
\newcommand{\Li}{\mathcal{K}}
\newcommand{\Lin}{ {\mathcal{K}_n} }

\newcommand{\M}{\mathcal{M}}
\newcommand{\cE}{\mathcal{E}}
\newcommand{\cR}{\mathcal{R}}

\newcommand{\wq}{\widetilde{q}}
\newcommand{\sL}{\xi}
\renewcommand{\wp}{\widetilde{p}}
\renewcommand{\eps}{\varepsilon}
\renewcommand{\leq}{\leqslant}
\renewcommand{\geq}{\geqslant}

\graphicspath{{/rapport/}}
\graphicspath{{/rapport/figures/}}

\begin{document}

  
\title{Langevin dynamics with space-time periodic nonequilibrium forcing}
\author{R. Joubaud,$^{1}$ G. Pavliotis$^{1}$ and G. Stoltz$^{2}$ \\
{\small $^{1}$ Department of Mathematics, Imperical College London, SW7 2AZ, London, UK} \\
{\small $^{2}$ Universit\'e Paris-Est, CERMICS (ENPC), INRIA, F-77455 Marne-la-Vall\'ee, France} \\}
 
\maketitle

\abstract{We present results on the ballistic and diffusive behavior of the Langevin dynamics in a periodic potential that is driven away from equilibrium by a space-time periodic driving force, extending some of the results obtained by Collet and Martinez in~\cite{ColletMartinez}. In the hyperbolic scaling, a nontrivial average velocity can be observed even if the external forcing vanishes in average. More surprisingly, an average velocity in the direction opposite to the forcing may develop at the linear response level -- a phenomenon called negative mobility. The diffusive limit of the non-equilibrium Langevin dynamics is also studied using the general methodology of central limit theorems for additive functionals of Markov processes. To apply this methodology, which is based on the study of appropriate Poisson equations, we extend recent results on pointwise estimates of the resolvent of the generator associated with the Langevin dynamics. Our theoretical results are illustrated by numerical simulations of a two-dimensional system.}

\tableofcontents

\section{Introduction}

Nonequilibrium transport has attracted a lot of attention in recent years both in the mathematical physics and physics literature. It is by now well understood that simple low dimensional systems that are driven away from equilibrium can exhibit quite complicated behavior such as stochastic resonance~\cite{Hanggi_al_1998}, directed transport~\cite{reimann}, absolute negative mobility~\cite{Hanggi_NegMob} and giant enhancement of diffusion~\cite{reimann_al01}. In particular, simple stochastic differential equations (SDEs) with space-time periodic coefficients have been proposed in recent years as models for Brownian motors. A natural question is whether SDEs with space-time mean zero coefficients can give rise to a nonzero effective drift. This question was studied in detail in~\cite{ColletMartinez}. The goal of this paper is to refine and extend the results obtained in this paper.

The long time dynamics of SDEs with space-time periodic coefficients is characterized by an effective drift, whereas fluctuations around the macroscopic directed transport are described by an effective Brownian motion with  covariance matrix $D$. The rigorous mathematical analysis of such models is based on proving the law of large numbers (ergodic theorem) at the hyperbolic timescale~\cite{ColletMartinez} and a functional central limit theorem~\cite{pavl05} at the diffusive timescale. These problems are closely related to the theory of homogenization for parabolic PDEs (and of the corresponding stochastic differential equations) with space-time dependent coefficients~\cite[Chapter~3]{BLP11}. Homogenization problems for Brownian motion in a space-time periodic potential were also studied in~\cite{garnier}.

On the other hand, nonequilibrium perturbations of systems at equilibrium can be used for calculating transport coefficients using linear response theory and the Green-Kubo formalism~\cite{KuboTodaHashitsume91,ResibDeLeen77}. Linear response theory can be rigorously analyzed for certain stochastic systems~\cite{Komor_Olla02, lebo_einstein, JS12, rodenh}, and higher order corrections can also be obtained~\cite{LatPavlKram2013}. The rigorous analysis of linear response theory and of the Green-Kubo formalism is very closely related to homogenization theory, in particular the systematic use of Poisson equations for obtaining formulas for transport coefficients. As an example we mention~\cite{JS12} where the linear response theory/Poisson equation formalism for the Langevin dyanamics is used to calculate the shear viscosity coefficient.

Stochastic dynamics in a bistable potential under the influence of an external time periodic forcing can exhibit stochastic resonance, whereby the stochastic trajectories are tuned in an optimal way to the deterministic forcing~\cite{Imkeller2005}. Frequency resonance can also appear for periodic potentials and for the underdamped Langevin dynamics, and it can be used in order to optimize the effective drift.

\bigskip

In this paper we consider perturbations of the equilibrium Langevin dynamics, obtained by adding a space-time periodic external force. We consider a system evolving in a periodic medium. The microscopic configuration is $(q,p) \in \M \times \mathbb{R}^d$, where $\M$ is the unit cell of some  periodic lattice (for simple cubic lattices, $\M = (L\mathbb{T})^d$ where $\mathbb{T} = \mathbb{R}\backslash \mathbb{Z}$ is the unit torus in dimension~one). 
The equilibrium Langevin dynamics reads
\begin{equation}
\label{eq:Langevin_eq}
\left\{ 
\begin{aligned}
dq_t & = M^{-1} p_t \, dt, \\
dp_t & = -\nabla V(q_t) dt - \gamma M^{-1} p_t \, dt + 
\sqrt{\frac{2\gamma}{\beta}} \, dW_t.
\end{aligned}
\right.
\end{equation}
We will consider the nonequilibrium dynamics obtained by adding a space-time periodic driving force~: 
\begin{equation}
\label{eq:Langevin}
\left\{ 
\begin{aligned}
dq^\eta_t & = M^{-1} p^\eta_t \, dt, \\
dp^\eta_t & = \Big( -\nabla V(q^\eta_t) + \eta F(t,q^\eta_t) \Big) dt - \gamma M^{-1} p^\eta_t \, dt + \sqrt{\frac{2\gamma}{\beta}} \, dW_t.
\end{aligned}
\right.
\end{equation}
In these equations, $\gamma > 0$ is the friction coefficient, $W_t$ is a standard $d$-dimensional Brownian motion, the mass matrix $M$ is a positive definite $d \times d$ matrix, and $\beta$ is inversely proportional to the temperature. Throughout this work, we will assume that

\begin{assumption}
The potential $V:\mathcal{M}\to\mathbb{R}$ and the external force $F:T\mathbb{T}\times \mathcal{M}\to\mathbb{R}$ are smooth and $\M$-periodic, and the external force is also time dependent and periodic with period~$T$. 
\end{assumption}

The dynamics~\eqref{eq:Langevin} has been studied in detail for constant external forcings~\cite{rodenh,LatPavlKram2013} and for particular types of space-dependent forcings~\cite{JS12}. In this paper we will focus on spatiotemporal periodic external forcings. It can be proved that the dynamics~\eqref{eq:Langevin} has a well-defined steady state (see Propositions~\ref{prop:inv_meas_cv} and~\ref{prop:inv_meas_charact}). When $\eta = 0$ (so that the dynamics is~\eqref{eq:Langevin_eq}), this steady-state is given by the canonical measure 
\begin{equation}\label{e:gibbs}
\mu(q,p) \, dq \, dp = \frac{1}{Z} \mathrm{e}^{-\beta H(q,p)} \, dq \, dp,
\end{equation}
where 
\begin{equation*}
H(q,p) = V(q) + \frac12 p^T M^{-1} p
\end{equation*}
is the Hamiltonian of the system. 

When the external force is constant, the systematic driving manifests itself through some average nonzero velocity in the system. When the external force is non-constant, and in particular when its space-time average vanishes, it is unclear whether a nontrivial average velocity can be observed. The surprising result by Collet and Martinez~\cite{ColletMartinez} is that, in fact, there is in general a nonzero average velocity even if the space-time average of the external force is zero. Our aim in this article is to refine and complete the results obtained in~\cite{ColletMartinez} where the case $V=0$ was studied. In particular, we present a more detailed analysis of the problem of convergence to equilibrium for the dynamics~\eqref{eq:Langevin}, we show that this system can exhibit the phenomenon of absolute negative mobility, we analyze the phenomenon of mobility resonance for~\eqref{eq:Langevin} and we prove a functional central limit theorem (homogenization theorem) for the particle position, under the diffusive rescaling. Our theoretical results are supported by numerical simulations.

\subsection*{Main results and organization of the paper}

The main results of this work are the following. 
\begin{enumerate}[(1)]
\item Proposition~\ref{prop:LR_velocity} and the reformulation~\eqref{e:aver-drift} show that the linear response of the average velocity is generically nontrivial even if the external force vanishes in average (provided its time-average is non-gradient, see Remark~\ref{rem:gradient}). The results we obtain extend the ones presented in~\cite{ColletMartinez}, by taking into account a nonzero potential~$V$ and giving explicitly the expression of the average time dependent velocity in the system. 

\item Upon adding an appropriate constant force we are also able to find situations in which the average velocity and the average external force experienced by the system are in oppposite directions (see Section~\ref{sec:neg_mob}). This phenomenon therefore corresponds to a situation of negative mobility. We emphasize the fact that we observe negative mobility at the level of linear response. In the existing literature negative mobility has been observed at the nonlinear response level, for time-independent forcings; in particular, a subtle interplay between periodic and static forcings at low temperatures were needed for this effect to be observed~\cite{Hanggi_NegMob}.

\item Resonance effects for the average drift in~\eqref{eq:Langevin} are studied in Section~\ref{sec:stoch_res}. We give mathematical properties of the amplitude of the time dependent response as a function of the period of the forcing, and present numerical results illustrating the phenomenon of resonance. To the best of our knowledge, these are the first results on such resonance effects for Langevin dynamics, while there are plenty of studies on the resonance of the average drift for overdamped Langevin dynamics, see for instance~\cite{Hanggi_al_1998}.

\item The effective diffusion obtained in a diffusive space-time scaling is studied in Section~\ref{sec:diffusive}, for arbitrary forcing magnitudes~$\eta$. We show in particular that the effective diffusion matrix varies at second order in~$\eta$ when the time average of the forcing is~0. 
\end{enumerate}

The rest of the paper is organized as follows. After studying the convergence of the dynamics to its (time dependent) stationary state in Section~\ref{sec:cv} for arbitrary perturbations~$\eta$, we prove in Section~\ref{sec:LR} various results on the linear response of the average velocity, and finally consider the diffusive regime in Section~\ref{sec:diffusive}. The proofs of the results presented in Sections~\ref{sec:LR} to~\ref{sec:diffusive} are gathered in Section~\ref{sec:proofs}. Numerical simulations for a simple two-dimensional potential, used to illustrate our findings, are presented throughout the paper.

\section{Convergence towards the nonequilibrium steady state}
\label{sec:cv}

We first introduce some notation. We denote by 
\begin{equation}\label{e:phase-space}
\mathcal{E} = T \mathbb{T} \times \mathcal{M} \times \RR^d
\end{equation} 
the extended phase-space. The (time dependent) generator of the process~\eqref{eq:Langevin} is $\cA_0 + \eta \cA_1$, with
\[
\cA_0 = M^{-1} p \cdot \nabla_q -\nabla V \cdot \nabla_p + \gamma 
\left( -M^{-1} p \cdot \nabla_p + \frac1\beta \Delta_p \right),
\qquad
\cA_1 = F(t,q) \cdot \nabla_p.
\]
We denote by $\cA_0^\dagger$ and $\cA_1^\dagger$ the adjoints of the generators (i.e. Fokker-Planck operators) on the space $L^2(\mathcal{E})$. We also introduce the family of Lyapunov functions for $n \geq 1$,
\[
\Li_n(q,p) = 1 + |p|^{2n},
\]
the associated weighted $L^\infty$ norms on functions $f(q,p)$ of $\mathcal{M} \times \RR^d$:
\[
\| f \|_{L^\infty_{\Li_n}} = \left\| \frac{f}{\Li_n} \right\|_{L^\infty},
\]
and the corresponding $L^\infty$ norms on functions $f(t,q,p)$ defined on~$\mathcal{E}$:
\[
\| f \|_{L^\infty(L^\infty_{\Lin})} = \sup_{\theta \in T\mathbb{T}} \| f(\theta) \|_{L^\infty_{\Li_n}}.
\]
Finally, for an element $t\in \RR$, we denote by $[t]$ the unique element of $[0,T)$ such that $t-[t] \in T\mathbb{Z}$, i.e. the value of $t$ modulo the period~$T$.

\bigskip

The first convergence result shows that the process stabilizes around a limit cycle described by a time dependent (periodic) invariant measure.

\begin{prop}[Uniform convergence to a limit cycle]
\label{prop:inv_meas_cv}
Fix $\eta_* > 0$ and $n \geq 1$. There exists a unique probability measure $\psi_\eta(\theta,q,p)$ on $T \mathbb{T} \times \mathcal{M} \times \RR^d$ and constants $C_n,\lambda_n > 0$ (depending on~$n$ and~$\eta_*$) such that, for any initial distribution~$(q_0,p_0)$ and for any $\eta \in [-\eta_*,\eta_*]$,
\begin{equation}
\label{eq:cv_exp_time_inhomog}
\forall f \in L^\infty(L^\infty_\Lin), \qquad \left| \mathbb{E}\Big( f([t],q^\eta_t,p^\eta_t)\Big) - \overline{f}_\eta([t])\right| \leq C_n \rme^{-\lambda_n t} \, \| f \|_{L^\infty(L^\infty_{\Lin})},
\end{equation}
where, for $\theta \in T\mathbb{T}$, the spatial average of~$f$ reads
\begin{equation}
\label{eq:local_time_average}
\overline{f}_\eta(\theta) = \int_{\mathcal{M} \times \RR^d} f(\theta,q,p) \, \psi_\eta(\theta,q,p) \, dq \, dp.
\end{equation}
The invariant distribution is smooth, positive ($\psi_\eta(t,q,p) > 0$ for all $(t,q,p) \in \mathcal{E}$) and satisfies the Fokker-Planck equation
\begin{equation}
  \label{eq:FokkerPlanck}
  \left(-\partial_t + \cA_0^\dagger + \eta \cA_1^\dagger\right) \psi_\eta = 0,
  \qquad
  \int_\cE \psi_\eta = 1.
\end{equation}
Finally, it has has finite moments of order~$2n$ uniformly in the time variable
\begin{equation}
\label{eq:finite_moments}
\forall \theta \in T \mathbb{T}, \qquad \int_{\mathcal{M} \times \RR^d} \Li_n(q,p) \, \psi_\eta(\theta,q,p) \, dq \, dp \leq R_n < +\infty,
\end{equation}
and has uniform marginals in the time variable:
\[
\overline{\psi_\eta}(\theta) = \int_\mathcal{E} \psi_\eta(\theta,q,p) \, dq \, dp = \frac1T.
\]
\end{prop}

Upon averaging in time, standard convergence results can be recovered, such as the following Law of Large Numbers, which will prove useful to study the effective diffusive behavior.
\begin{prop}
\label{prop:LLN}
Consider $\eta \in \RR$ and $f \in L^\infty(L^\infty_\Lin)$. Then, for any initial condition $(q_0,p_0)$,
\begin{equation}
\label{eq:LLN}
\frac1t \int_0^t f([s],q^\eta_s,p^\eta_s) \, ds \xrightarrow[t \to +\infty]{} \int_\mathcal{E} f \, \psi_\eta \qquad \mathrm{a.s.}
\end{equation}
\end{prop}

The invariant measure can be fully characterized in the linear response regime as a perturbation around the equilibrium measure~$\mu$ defined in~\eqref{e:gibbs}, for forcings sufficiently small. Let us emphasize that this result is perturbative, in contrast to the convergence statement given by Proposition~\ref{prop:inv_meas_cv}. Similar results have been obtained for different stochastic systems in~\cite{KomOlla2005}.

\begin{prop}[Series expansion of the invariant measure for small forcings]
\label{prop:inv_meas_charact}
There exists $C,r > 0$ such that, for $|\eta| < r$, the invariant measure is given by the following series expansion in~$\eta$: 
\[
\psi_\eta(t,q,p) = \rho_\eta(t,q,p) \mu(q,p), 
\qquad 
\rho_\eta(t,q,p) = 1 + \eta \varrho_1(t,q,p) + \eta^2 \varrho_2(t,q,p) + \dots
\]
with 
\begin{equation}
\label{eq:conditions_varrho}
\int_\mathcal{E} |\varrho_m(t,q,p)|^2 \mu(q,p)\,dq\,dp\,dt \leq \frac{C}{r^m}, 
\qquad
\int_\cE \varrho_m(t,q,p) \, \mu(q,p)\,dq\,dp\,dt = 0.
\end{equation}
\end{prop}
The functions $\varrho_m$ are not explicitely known, but are defined as solutions of appropriate Poisson equations (see Section~\ref{sec:inv_meas_proof_charact}). The leading order correction $\varrho_1$ is particularly important since it governs the linear response.

\section{Linear response of the velocity}
\label{sec:LR}

\subsection{General result}
\label{sec:general_LR}

For a given perturbation strength, define the time dependent spatially averaged velocity 
\[
\overline{v}_\eta(t) = \int_\M \int_{\mathbb{R}^d} M^{-1}p \, \psi_\eta(t,q,p) \, dq \, dp
\]
for any $t \in [0,T]$, and the associated linear response
\[
\mathscr{V}(t) = \lim_{\eta \to 0} \frac{\overline{v}_\eta(t)}{\eta}.
\]
To decompose $\mathscr{V}(t)$, we introduce the (unnormalized) Fourier modes on $L^2(T \mathbb{T})$
\[
e_n(t) = \mathrm{e}^{\ri n \omega t}, \qquad \omega = \frac{2\pi}{T},
\]
and first decompose the real-valued external force as
\[
F(t,q) = F_0(q) + \sum_{n \in \mathbb{Z} \backslash \{ 0 \} } F_n(q) e_n(t) = 
F_0(q) + 2\sum_{n \geq 1} \mathrm{Re}\Big( F_n(q) e_n(t) \Big),
\]
with
\[
F_n(q) = \frac1T \int_0^T F(t,q)\,  \mathrm{e}^{-\ri n \omega t}\,dt.
\]
Note that $F_{-n} = \overline{F_n}$ (the bar indicating here complex conjugation).

The following result (proved in Section~\ref{sec:proof_LR_velocity}) shows that each time harmonic of the linear response of the time dependent velocity is directly proportional to the corresponding harmonic of the external force. We will use the notation
\[
\widetilde{\mu}(q) = \int_{\mathbb{R}^d} \mu(q,p) \, dp = \widetilde{Z}^{-1} \mathrm{e}^{-\beta V(q)} 
\]
for the marginal density of the canonical measure in the position variables ($\widetilde{Z}$ denotes the normalization constant), and, for a given operator~$A$, consider the element $Ap$ as the vector with components $A p_i$. 

\begin{prop}
\label{prop:LR_velocity}
The linear response of the time dependent spatially averaged velocity can be related to the external force as
\begin{equation}
\label{eq:full_LR}
\mathscr{V}(t) = \beta \sum_{n \in \bZ} e_n(t) \int_\M D_n(q) F_{n}(q)\, \widetilde{\mu}(q) \, dq, 
\end{equation}
where the position-dependent diffusion matrix reads
\begin{align}
D_n(q) & = \int_0^{+\infty} \mathbb{E}\Big( \left(M^{-1}p_s\right) \otimes \left(M^{-1}p_0\right) \, \Big| \, q_0 = q\Big) \mathrm{e}^{\ri n \omega s} ds \nonumber \\
& = -\left(\frac{2\pi}{\beta}\right)^{-d}|\mathrm{det}(M)|^{-1/2}\int_{\mathbb{R}^d} \left[(\ri n\omega+\cA_0)^{-1} \left(M^{-1}p\right)\right] \otimes \left(M^{-1}p\right) \, \exp\left(-\beta\frac{p^T M^{-1}p}{2}\right)\,dp, \label{eq:reformulation_Dn}
\end{align}
the expectation in the first equality being with respect to canonically distributed initial momenta $p_0$, and for all realizations of the equilibrium Langevin dynamics~\eqref{eq:Langevin_eq} starting from $(q,p_0)$. In particular, the average (time-independent) velocity depends only on the component $F_0$ of the external force:
\begin{equation}
\label{eq:time_avg_response}
\overline{\mathscr{V}} = \frac{1}{T} \int_0^T \mathscr{V}(t) \, dt
= \beta \int_\M D_0(q) F_0(q) \, \widetilde{\mu}(q) \, dq.
\end{equation}
\end{prop}

\begin{remark}
Proposition~\ref{prop:LR_velocity} refines the results of~\cite{ColletMartinez} in two ways: (i) it gives an expression of $\mathscr{V}(t)$ and not only of its time average, and (ii) it highlights the fact that $F_0$ solely determines whether the average velocity vanishes or not. 
\end{remark}

\begin{remark}[The average velocity vanishes for gradient perturbations]
\label{rem:gradient}
When $F(t,q) = F_0(q) = -\nabla W(q)$, the process has an invariant measure whose explicit expression is known: it is the canonical measure associated with the potential energy function $V+\eta W$. In this case, the average velocity should vanish. In fact, the average velocity $\overline{\mathscr{V}}$ is zero as soon as the time-averaged external force is given by the gradient of a scalar function: $F_0(q) = -\nabla W(q)$, as can be seen from~\eqref{eq:time_avg_response}. Indeed, with expectations taken for all initial conditions distributed according to the equilibrium steady state~$\mu$ defined in~\eqref{e:gibbs} and for all realizations of the equilibrium Langevin dynamics~\eqref{eq:Langevin_eq},
\begin{align*}
\overline{\mathscr{V}} & = -\beta \int_0^{+\infty}
\mathbb{E}\left[  \left(M^{-1}p_t \otimes M^{-1}p_0 \right) \nabla W(q_0) \right] \, dt \\
& = -\beta \int_0^{+\infty}
\mathbb{E}\left[  \left(M^{-1}p_0 \otimes M^{-1}p_t \right) \nabla W(q_t) \right] \, dt \\
& = -\beta \, \mathbb{E}\left[M^{-1}p_0 \int_0^{+\infty} \left(M^{-1}p_t\right)^T \nabla W(q_t) \, dt\right].
\end{align*}
where we have used the time-reversal invariance to go from the first to the second line (namely $\cA_0^* = \cR \cA_0 \cR$ where $\cR \varphi(q,p) = \varphi(q,-p)$, hence $(\rme^{t\cA_0})^* = \cR \rme^{t \cA_0} \cR$). We now use the fact that, for any $\tau > 0$,
\[
\mathbb{E}\left[M^{-1}p_0 \int_0^\tau \left(M^{-1}p_t\right)^T \nabla W(q_t) \, dt\right] = \mathbb{E}\left[M^{-1}p_0 W(q_\tau)\right]-\mathbb{E}\left[M^{-1}p_0W(q_0))\right].
\]
The second expectation vanishes. For the first expectation we use the fact that, by ergodicity, the law of $q_\tau$ has some limiting behavior whatever the choice of~$p_0$, so that
\[
\lim_{\tau \to +\infty} \mathbb{E}\left[M^{-1}p_0 W(q_\tau)\right] = \mathbb{E}\left[M^{-1}p_0 \left(\lim_{\tau \to +\infty} \mathbb{E}\left[\left. W(q_\tau)\, \right| \, \mathcal{F}_0\right] \right) \right] = 0.
\]
\end{remark}

\subsubsection*{Numerical illustration}

We illustrate the results from Proposition~\ref{prop:LR_velocity} with some numerical results.
We consider a single particle of mass~$1$ in the two-dimensional potential
\begin{equation}\label{eq:pot2d}  
V(q) = 2\cos(2x)+\cos(y) + \cos(x-y),
\end{equation}
where $q = (x, \, y)$. The numerical scheme is obtained by a Strang splitting between the Hamiltonian part and the fluctuation-dissipation part (including the nonequilibrium forcing). More precisely, denoting by $(q^n,p^n)$ approximations of $(q_{n \Delta t},p_{n \Delta t})$ (to simplify the notation, we do not explicitly denote the dependence on~$\eta$), 
\begin{equation}
\label{eq:num_Langevin}
\left\{ 
\begin{aligned}
p^{n+1/2} & = \alpha p^n +\frac{\Delta t}{2} \Big(-\nabla V(q^n)+ \eta F(t^n,q^n)\Big)
+ \sqrt{\frac{1-\alpha^2}{\beta}} \, G^n, \\
q^{n+1} & = q^n + \Delta t \, M^{-1} p^{n+1/2}, \\
p^{n+1} & = \alpha p^{n+1/2} +\frac{\Delta t}{2} \Big(-\nabla V(q^{n+1})+ \eta F(t^{n+1},q^{n+1})\Big) 
+ \sqrt{\frac{1-\alpha^2}{\beta}} \, G^{n+1/2}, \\
\end{aligned}
\right.
\end{equation}
with $\alpha = \mathrm{e}^{-\gamma \Delta t \, M^{-1}/2}$ and where $G^n,G^{n+1/2}$ are i.i.d. 2-dimensional Gaussian random vectors. We refer to~\cite{LMS13} for a numerical analysis of the errors on the linear responses computed using this numerical scheme.

The average time dependent velocities $\mathscr{V}(t)$ are approximated at the times
$\tau_i = i\Delta t \in [0,T]$, for $i=1,\dots,I$ with $I \Delta t = T$.
These approximations are numerically computed as longtime averages 
over trajectories $(q^n,p^n)_{n=1,\dots,N}$ of~\eqref{eq:num_Langevin} (with $N/I \in \mathbb{N}$) by considering velocities only at times $\tau_i, \tau_i + T, \tau_i + 2T, \dots$:
\begin{equation}
  \label{eq:v_eta_num}
  \overline{v}_\eta(\tau_i) \simeq \frac{I}{N} \sum_{j=1}^{N/I} M^{-1} p^{i + jI}.
\end{equation}
The average, time-independent linear response $\overline{\mathscr{V}}$ is then approximated by fitting the linear dependence of the average velocity
\[
\frac{1}{I} \sum_{i=1}^I \overline{v}_\eta(\tau_i)
\]
as a function of $\eta$ using a least-square fit.

In the numerical experiments reported below, the numerical parameters are set to $\Delta t = 0.01$, $\beta = 1$, $\gamma = 1$, and the dynamics was integrated over $N = 4.5 \times 10^9$ time-steps. The maximum value of the forcing strength is $\eta_{\rm max}=1$. We performed $R=100$ independent simulations with equally spaced intermediate values $\eta_{\rm max}/R, 2\eta_{\rm max}/R, \dots, \eta_{\rm max}$. 

We consider non-gradient external forcings of the form
\[
F_{0,n}(q) = \mathrm{e}^{\beta V(q)} \begin{pmatrix}
\cos(nx) \\ 
0 \\ 
\end{pmatrix},
\]
with $n\in \bN$. The results are presented in Figure~\ref{fig:SLR}. We tested the values $n=1,2,3,4$. Nontrivial responses are obtained for all of them. In this simple example, the forcings in the $x$ direction induce a first-order response in the $x$ and $y$ directions (not documented in the pictures of Figure~\ref{fig:SLR}; the response in the $y$ direction roughly is 5 times smaller than the response in the $x$ direction). 

\begin{figure}[h]
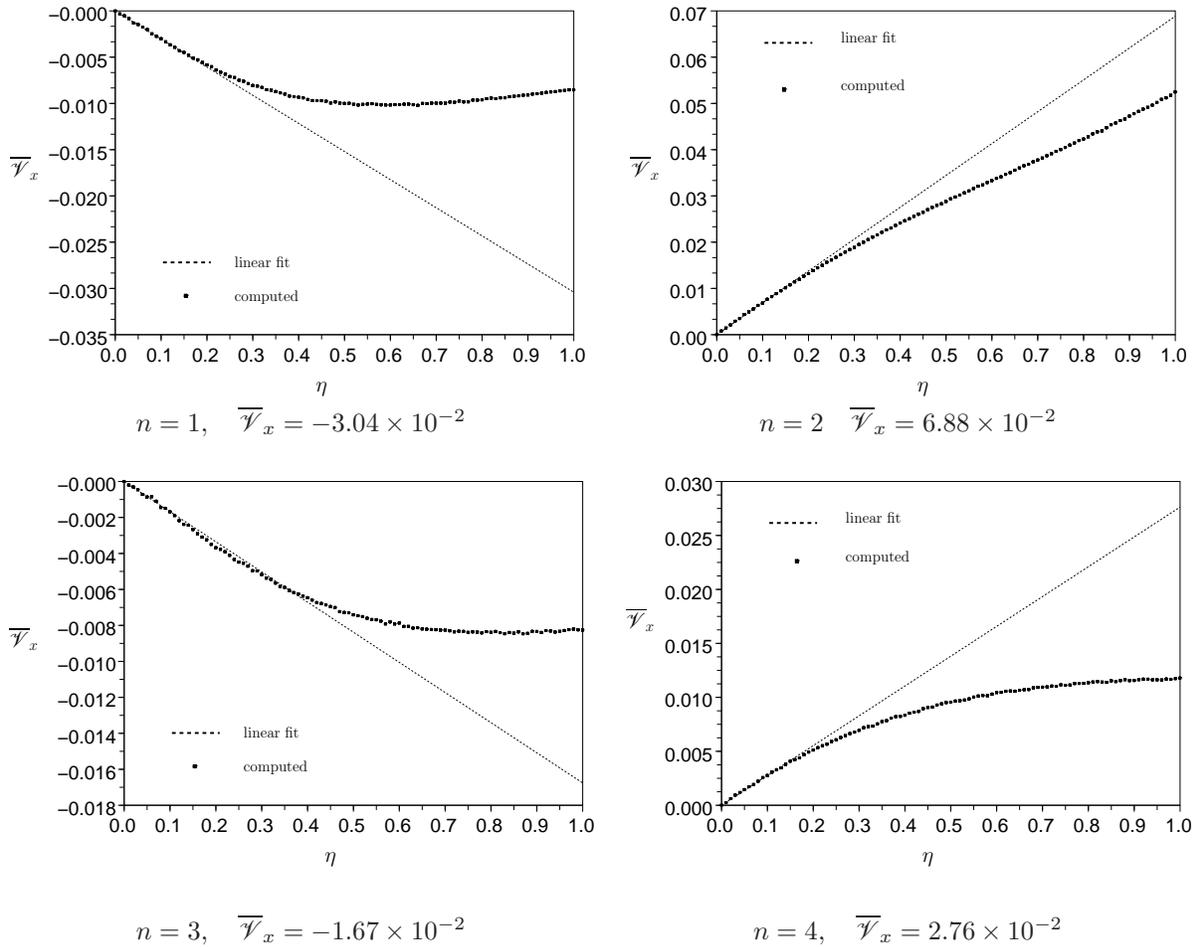

\begin{tabular}{cc}
 \scalebox{0.4}{\input{LR_1x.pstex_t}}&\scalebox{0.4}{\input{LR_2x.pstex_t}}\\
\vspace{0.5cm}

$n=1,\quad\overline{\mathscr{V}}_x = -3.04\times 10^{-2}$& $n=2\quad \overline{\mathscr{V}}_x = 6.88\times 10^{-2}$\\
\vspace{0.5cm}

 \scalebox{0.4}{\input{LR_3x.pstex_t}}&\scalebox{0.4}{\input{LR_4x.pstex_t}}\\
\vspace{0.5cm}

$n=3,\quad\overline{\mathscr{V}}_x = -1.67\times 10^{-2}$&$n=4,\quad \overline{\mathscr{V}}_x = 2.76\times 10^{-2}$ 
\end{tabular}
\caption{\label{fig:SLR} Average velocity as a function of the applied force for the various forcings.}
\end{figure}

\subsection{Negative mobility}
\label{sec:neg_mob}

The physical interpretation of~\eqref{eq:time_avg_response} is that a nontrivial velocity can be observed when the spatial modes of the position-dependent diffusion matrix $D_0(q)$ are excited by the external forcing. The aim of this section is to make this observation precise, and use it to construct situations in which a negative mobility is observed.

We can rewrite $D_0$ as
\begin{equation}
  \label{eq:expression_D0}
  D_0(q) = \int_0^{+\infty} \mathbb{E}\Big( \left(M^{-1}p_s\right) \otimes \left(M^{-1}p_0\right) \, \Big| \, q_0 = q\Big) ds = \int_{\RR^d} \Phi_0(q,p) \otimes \Big(M^{-1}p \Big) g(p) \, dp,
\end{equation}
where $g(p)$ is the density of the Gaussian distribution $(2\pi)^{-d/2} \mathrm{det}(M)^{1/2} \exp(-\beta p^T M^{-1} p/2)$ and $\Phi_0$ is the unique solution of the Poisson equation 
\begin{equation}
\label{eq:Phi_0}
-\cA_0 \Phi_0 = M^{-1} p, \qquad \int_\cE \Phi_0(q,p) \, \mu(q,p) \, dq \, dp = 0.
\end{equation}
Note that the function~$\Phi_0$ is time-independent. A priori we should consider the above equation as posed on~$\cE$ and look for $\Phi_0 \in (L^2(\cE ; \mu))^d$, but the existence of a time-independent solution (as given by Lemma~\ref{lem:pt_A0_bounded}) and the uniqueness of solutions, as given by Lemma~\ref{lem:unif_resolvent_L2}, enable us to conclude that the time-dependence can be removed.

Since the matrix $D_0$ is real, $\M$-periodic and symmetric (by the same time-reversal argument as in Remark~\ref{rem:gradient}), it can be written as
\begin{equation}
\label{eq:decomposition_D0}
D_0(q) = \sum_{K \in \mathcal{L}^*} D_{0,K} \mathrm{e}^{-\ri K \cdot q} 
= \sum_{K \in \mathcal{L}^*} a_{0,K} \cos(K \cdot q) + b_{0,K} \sin(K \cdot q), 
\end{equation}
where $\mathcal{L}^*$ is the reciprocal lattice associated with the lattice $\mathcal{L}$ whose unit cell is~$\M$, and $a_{0,L},b_{0,L}$ are real, symmetric $d \times d$ matrices. The following result (see Section~\ref{sec:proof_genuine_spatial_dep} for the proof) shows that there should be in general a non-trivial $q$-dependence in~\eqref{eq:decomposition_D0} since, in view of~\eqref{eq:expression_D0},
\[
\nabla_q D_0(q) = \int_{\RR^d} \nabla_q \Phi_0(q,p) \otimes \Big(M^{-1}p \Big) g(p) \, dp.
\]
When the average on the right-hand side is non-zero, it indeed holds $D_0(q) \neq D_{0,0}$.

\begin{prop}
\label{prop:genuine_spatial_dep}
Assume that the following non-degeneracy condition holds:
\begin{equation}
\label{eq:non_degeneracy}
\mathrm{Span}\left \{ \nabla V(q), \, q \in \mathcal{M} \right\} = \mathbb{R}^d. 
\end{equation}
Then, the smooth function~$\Phi_0 = (\Phi_{0,1},\dots,\Phi_{0,d})$ has a genuine spatial dependence, \textit{i.e.} $\nabla_q \Phi_{0,i} \neq 0$ for all $i \in \{ 1,\dots,d\}$. 
\end{prop}

The non-degeneracy condition~\eqref{eq:non_degeneracy} ensures that the conservative force is sufficiently rich to obtain a coupling between the potential and the driving force. It prevents the use of trivial potential such as $V(q) = v(q_1)$. 

\bigskip 

Upon decomposing $F_0 \widetilde{\mu}$ in Fourier series, the time-averaged external force $F_0$ can be written in general as 
\[
F_0(q) = \frac{1}{\widetilde{\mu}(q)|\M|} \left( \sum_{K \in \mathcal{L}^*} F_{0,K} \mathrm{e}^{-\ri K \cdot q}\right), \qquad F_{0,K} = \overline{F_{0,-K}} \in \mathbb{C}^{d \times d}.
\]
The normalization is chosen such that $F_{0,0}$ is the canonical average of $F_0$. 
Mean zero forces are characterized by the condition
\[
  F_{0,0} = \int_\M F_0(q) \widetilde{\mu}(q) \, dq = 0.
\]
The space-time averaged linear response of the velocity reads, in view of~\eqref{eq:time_avg_response},
\begin{equation}\label{e:aver-drift}
\overline{\mathscr{V}} = D_{0,0} F_{0,0} + \int_\M \Big( D_0(q)-D_{0,0} \Big)\Big( F_0(q) \widetilde{\mu}(q) - F_{0,0}\Big)dq.
\end{equation}
This decomposition highlights the two mechanisms separately contributing to the mean velocity: (i) the response to a constant forcing with average $F_{0,0}$, in which case the relevant diffusion matrix is a spatial average of the position dependent diffusion matrix $D_0(q)$; (ii) an additional contribution arising from a spatial resonance effect between two terms whose spatial averages vanish, namely the centered diffusion matrix $D_0(q) - D_{0,0}$ and the centered (with respect to the canonical measure) force $F_0(q) - F_{0,0} \widetilde{\mu}(q)^{-1}$.

Proposition~\ref{prop:genuine_spatial_dep} shows that it is possible to obtain a nontrivial velocity even when the external force vanishes on average, upon choosing an element $K_0 \in \mathcal{L}^* \backslash \{ 0 \}$ such that $D_{0,K_0} \neq 0$, and choosing for instance a forcing 
\[
F_0(q) = \mathrm{e}^{\beta V(q)} \cos(K_0 \cdot q) F_{0,K_0}, 
\qquad
D_{0,K_0} F_{0,K_0} \neq 0.
\]
In fact, the decomposition~\eqref{e:aver-drift} then shows how to obtain a situation with negative mobility: pick a function $F_0$ such that the second term on the righthand side of~\eqref{e:aver-drift} (denoted by $\overline{\mathscr{V}}_{\rm res}$) is nonzero, and add then a not too large constant force $F_{0,0}$ leading to a contribution $\overline{\mathscr{V}}_{\rm cst} = D_{0,0} F_{0,0}$ with 
\[
F_{0,0} \cdot \overline{\mathscr{V}} = F_{0,0} \cdot \Big( D_{0,0} F_{0,0} + \overline{\mathscr{V}}_{\rm res} \Big) < 0.
\]

\subsubsection*{Numerical illustration}

We now use the numerical algorithm presented in the previous section to study numerically a situation in which negative mobility at the linear response level appears. We consider the force
\begin{equation}
\label{eq:NM_force}
F_{0}(q) = \mathrm{e}^{\beta V(q)} \begin{pmatrix}
-1 + 3 \cos(2x) \\ 
0 \\
\end{pmatrix}.
\end{equation}
The parameters are chosen so that the linear response in the $x$ direction is positive, while the average force in the same direction is negative. The average external force actually experienced by the system is estimated over trajectories $(q^n,p^n)_{n=1,\dots,N}$ of~\eqref{eq:num_Langevin} as
\[
\frac1N \sum_{n=1}^N F_0(q^n),
\]
and we retain only the $x$ component (the $y$ component vanishes at first order in~$\eta$). The numerical results presented in Figure~\ref{fig:neg_mob} indeed confirm that the average experienced force and the associated observed velocities are of opposite signs.

\begin{figure}[h]
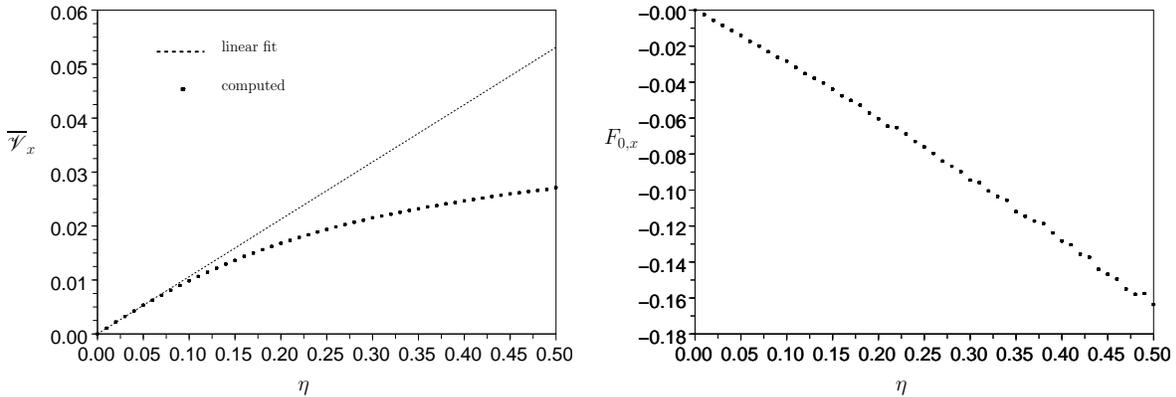

\begin{tabular}{cc}
 \scalebox{0.4}{\input{LR_NM.pstex_t}}&\scalebox{0.4}{\input{NM_force.pstex_t}}\\
\end{tabular}
\caption{\label{fig:neg_mob} Negative mobility for the force~\eqref{eq:NM_force}. Left: Average observed velocity in the $x$ direction. Right: Average experienced force in the $x$ direction.}
\end{figure}

\subsection{Resonance of the frequency-dependent mobility}
\label{sec:stoch_res}

We consider now the dependence of the effective velocity $\mathscr{V}(t)$ on the period~$T$ of the external forcing. We restrict ourselves to monochromatic forcings at frequency~$\omega = 2\pi/T$:
\[
F(t,q) = 2 \mathrm{Re}\left(F_1(q) \, \mathrm{e}^{\ri \omega t} \right).
\]
In the sequel, we fix the force $F_1(q) = \overline{F_{-1}}(q)$ and only vary the frequency~$\omega$. From~\eqref{eq:full_LR}, it is clear that the linear response $\mathscr{V}(t)$ involves a linear response at a single frequency~$\omega$:
\[
\mathscr{V}(t) = 2\beta \, \mathrm{Re} \left( e_1(t) \int_\M D_1(q) F_{1}(q)\, \widetilde{\mu}(q) \, dq \right) = 2\beta \, \mathrm{Re} \left( \widehat{\mathscr{V}}(\omega) \, \rme^{\ri \omega t} \right),
\]
with (recall~\eqref{eq:reformulation_Dn})
\[
\widehat{\mathscr{V}}(\omega) = -2\beta \int_\M \int_{\mathbb{R}^d} \Big( \left[(\ri \omega + \cA_0)^{-1} \left(M^{-1}p\right)\right] \otimes \left(M^{-1}p\right) \Big) F_{1}\,\mu.
\]
The resolvent $(\ri \omega + \cA_0)^{-1}$ in the latter expression shows that the magnitude of the response is the result of an interplay between the forcing period (related to the inverse of~$\omega$) and the typical relaxation time of the dynamics (related to the spectral gap of~$\mathcal{A}_0$).

Frequency resonance corresponds to a nonmonotonic behavior of the magnitude $\left|\widehat{\mathscr{V}}(\omega)\right|$ of the linear response, which has some (local) maximal value for some frequency $0 < \omega_{\rm res} < +\infty$. The existence of at least one maximizer in~$[0,+\infty)$ is related to the fact that $\widehat{\mathscr{V}}$ has a definite value at $\omega = 0$, and vanishes as $\omega \to +\infty$, as stated by the following result (proved in Section~\ref{sec:SR_proof}). 

\begin{prop}
\label{prop:SR}
For any $n \geq 2$, there exists a constant $C_n > 0$ and vectors $\nu_1,\dots,\nu_{n-1} \in \mathbb{C}^d$, such that, for all $\omega \geq 1$,
\[
\left| \widehat{\mathscr{V}}(\omega) - \sum_{m=1}^{n-1} \frac{\nu_m}{\omega^m} \right| \leq \frac{C_n}{\omega^n},
\]
with 
\[
\nu_1 = 2 \ri \beta M^{-1} \int_\M F_{1}(q) \, \widetilde{\mu}(q) \, dq.
\]
In particular, the vectorial amplitude $\widehat{\mathscr{V}}(\omega)$ vanishes as $\omega \to +\infty$.
\end{prop}

The interpretation of this result is that, when the forcing period is very small (\textit{i.e.} $\omega$ is very large), the system does not have time to follow the excitation provided by the external forcing, so that the observed response is very small. The above asymptotic result does not however give any information on the existence of a local maximum of $\left | \widehat{\mathscr{V}}(\omega) \right|$ for $\omega > 0$. We therefore have to rely on numerical simulations to study this effect.

\subsubsection*{Numerical illustration}

We use the notation introduced in the previous section. To obtain the amplitude of the linear response, we consider $R=100$ independent simulations with equally spaced intermediate values $\eta_{\rm max}/R, 2\eta_{\rm max}/R, \dots, \eta_{\rm max}$. For each simulation, we compute approximations of the time dependent average velocity $\overline{v}_\eta(t)$ in a given direction ($x$ or $y$) as given by~\eqref{eq:v_eta_num}, and perform a fast Fourier transform to extract the amplitude of the mode at frequency~$\omega$. The linear response of this quantity (computed using a least-square fit) gives the desired approximation of $\widehat{\mathscr{V}}(\omega)_x$ or $\widehat{\mathscr{V}}(\omega)_y$.

Figure~\ref{fig:resonance} presents results on the resonance of the mobility obtained in the $x$~direction for forcings of the form
\begin{equation}
\label{eq:force_SR}
F(t,q) = \mathrm{e}^{\beta V(q)} \begin{pmatrix} \cos(2x) \\ 0 \end{pmatrix} \cos(\omega t),
\end{equation}
with a friction $\gamma = 0.1$ and a maximal forcing strength $\eta_{\rm max} = 0.5$, the other numerical parameters being the same as in Section~\ref{sec:general_LR}. The time-step was refined to $\Delta t = 0.001$ instead of $\Delta t = 0.01$ for periods $T \leq 1$, and extended to $\Delta t = 0.025$ for $T \geq 20$. The computations show the existence of a resonance frequency around $\omega/(2\pi) \simeq 0.45$. By Proposition~\ref{prop:SR}, the decrease of the amplitude is at least of order $\omega^{-2}$ as $\omega \to +\infty$ since $\nu_1 = 0$. Such a fast decay is difficult to observe numerically.

\begin{figure}[h]
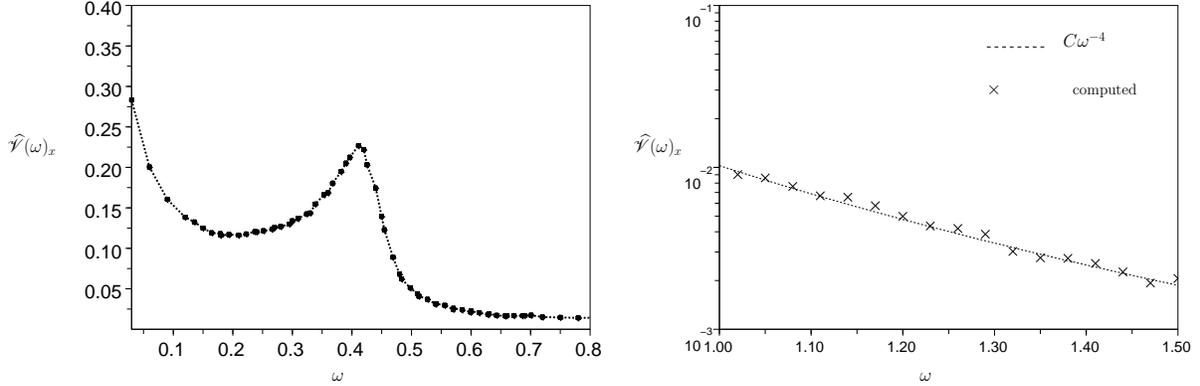

\begin{center}
\begin{tabular}{cc}
\scalebox{0.4}{\input{fr_cas1_pike.pstex_t}}&\scalebox{0.4}{\input{fr_cas1_asymptotic.pstex_t}}
\end{tabular}
\end{center}
\caption{\label{fig:resonance} Plot of $\widehat{\mathscr{V}}(\omega)_x$ as a function of the frequency~$\omega/(2\pi)$, for the force~\eqref{eq:force_SR}. The fit in log-log scales suggests that $\widehat{\mathscr{V}}(\omega)_x \sim \omega^{-4}$.}
\end{figure}

To observe more easily the predicted decay of the amplitude, we consider the force
\begin{equation}
\label{eq:force_dir}
F(t,q) = \begin{pmatrix} 1 \\ 0 \end{pmatrix} \cos(\omega t)
\end{equation}
for which $\nu_1 \neq 0$. The numerical results presented in Figure~\ref{fig:resonance_dir} show two interesting features: the resonance peak corresponds to a global maximum of the amplitude, and a decay of order $\omega^{-1}$ can be observed at large frequencies, as predicted by Proposition~\ref{prop:SR}. 

\begin{figure}[h]
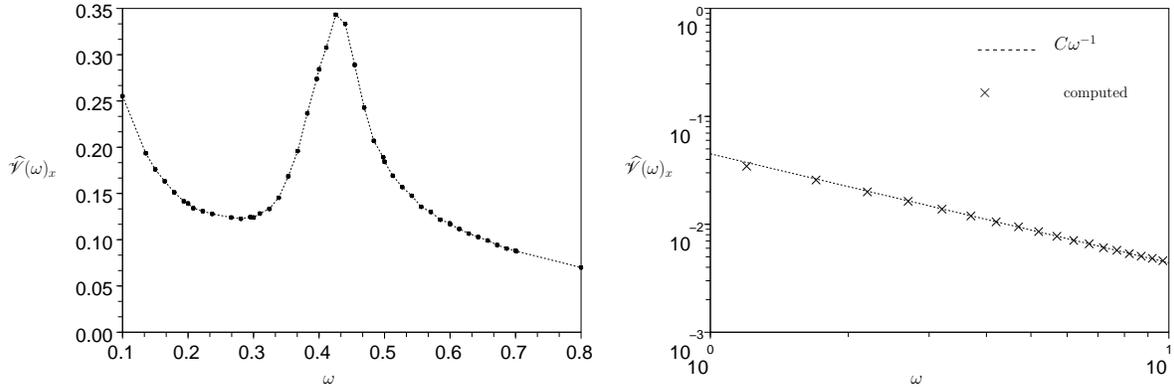

\begin{center}
\begin{tabular}{cc}
\scalebox{0.4}{\input{fr_cas2_pike.pstex_t}}&\scalebox{0.4}{\input{fr_cas2_asymptotic.pstex_t}}
\end{tabular}
\end{center}
\caption{\label{fig:resonance_dir} Plot of $\widehat{\mathscr{V}}(\omega)_x$ as a function of the frequency~$\omega/(2\pi)$, for the force~\eqref{eq:force_dir}. The fit in log-log scales suggests that $\widehat{\mathscr{V}}(\omega)_x \sim \omega^{-1}$.}
\end{figure}

\section{Longtime diffusive behavior}
\label{sec:diffusive}

In the section we study the longtime diffusive behavior of the nonequilibrium dynamics~\eqref{eq:Langevin}; in particular we show that the diffusively rescaled particle position converges weakly to a Brownian motion with an appropriate diffusion matrix. This result is based on the well developed techniques for proving functional central limit theorems for additive functionals of Markov processes~\cite{KS86, MasFerGoldWi89, Cattiaux-al2012, KomorowskiLandimOlla2012}, in particular on the study of an appropriate Poisson equation and on the use of the martingale central limit theorem.

\subsection{Bounds on the solution of the Poisson equation}

In order to properly define the diffusion matrix, the fundamental ingredient is estimates on the solution of appropriate Poisson equations. One possible result is the following, which gives a polynomial control on derivatives of arbitrary order.

\begin{prop}
\label{prop:Poisson}
Consider a smooth function~$f$ with derivatives growing at most polynomially in~$p$. Then, there exists a unique solution $\Phi_\eta$ to the Poisson equation (recall that the phase-space $\cE$ is introduced in~\eqref{e:phase-space})
\begin{equation}
\label{eq:Poisson_equation}
(\partial_t + \cA_0 + \eta \cA_1) \Phi_{\eta}(t,q,p) = f(t,q,p) - \int_\cE f \psi_\eta, 
\qquad
\int_\cE \Phi_\eta(t,q,p) \, \psi_\eta(t,q,p) \, dt \, dq \, dp = 0.
\end{equation}
Moreover, for any $k \geq 1$ and $\eta_* > 0$, there exists a real constant $C > 0$ and integers $n,m,N \geq 1$ such that, for all $\eta \in [-\eta_*,\eta_*]$,
\begin{equation}
\label{eq:pointwise_Poisson}
|\partial^l \Phi_\eta(t,q,p)| \leq C \Li_n(q,p) \sup_{\substack{ r\in \mathbb{N}^{2d} \\ |r| \leq N}} \left\| \partial^r f \right\|_{L^\infty(L^\infty_{\Li_m})},
\end{equation}
where $\partial^l$ with $l\in \mathbb{N}^{2d}$ is a derivative of order at most~$k$, \textit{i.e.} $\partial^l = \partial_{q_1}^{l_1} \dots \partial_{q_d}^{l_d} \partial_{p_1}^{l_{d+1}} \dots \partial_{p_d}^{l_{2d}}$ with $|l| = l_1 + l_2 + \dots + l_{2d} \leq k$.
\end{prop}

This result is proved in Section~\ref{sec:proof_Poisson}, as an extension of a technique first employed by Talay in~\cite{Talay02}, which was recently carefully rewritten and extended in~\cite[Appendix~A]{Kopec13}. We emphasize the fact that the equation~\eqref{eq:Poisson_equation} is equipped with periodic boundary conditions in both $t$ and $q$. The boundary condition for $p$ is that $\Phi_{\eta} \in L^2(\psi_{\eta})$. We also remark that we do not need precise information on the invariant measure $\psi_\eta$. In particular, it need not be a perturbation of the invariant measure of the equilibrium dynamics~$\mu$.

\subsection{Definition of the effective diffusion matrix}

The behavior of the process depends in the long time limit on the space-time scaling. In the hyperbolic scaling where time and space are renormalized by the same factor, the behavior is determined by the average velocity
\[
\mathcal{V}_\eta = \frac1T \int_0^T \overline{v}_\eta(t) \, dt = \int_{\mathcal{E}} M^{-1} p \, \psi_{\eta}(t,q,p) \, dt\, dq \, dp.
\]
In general, as discussed in Section~\ref{sec:LR}, this average velocity is not zero, and is of order~$\eta \overline{\mathscr{V}}$ when~$\eta$ is small (see~\eqref{eq:expansion_V_eta} below). The behavior of the process over longer times, in a diffusive space-time scaling, describes the deviations around the average velocity. We introduce for this study the process
\[
\mathscr{Q}_t^\eta = q_0^\eta + \int_0^t M^{-1} p^\eta_s \, ds. 
\]
The only difference between $q_t^\eta$ and $\mathscr{Q}_t^\eta$ is that $\mathscr{Q}_t$ is not reprojected in~$\mathcal{M}$ by the periodization procedure, and hence can diverge as time passes. In the long time limit, $\mathscr{Q}_t^\eta$ drifts linearly as $t \mathcal{V}_\eta$. The diffusive behavior is captured by the centered process
\[
Q_t^\eta = \mathscr{Q}_t^\eta - t\mathcal{V}_\eta
\] 
considered in a diffusive space-time scaling: for any $\eps > 0$, 
\[
Q_t^{\eta,\eps} = \eps Q^\eta_{t/\eps^2} = \eps \mathscr{Q}_{t/\eps^2}^\eta- \frac{t}{\eps} \mathcal{V}_\eta.
\]
We assume that the process starts at stationary initial conditions, namely $(q^\eta_0,p^\eta_0) \sim \psi_\eta(0,q,p)\,dq\,dp$. The stationarity assumption can be removed at the expense of additional technical difficulties. 
Introduce finally the following Poisson equation (well defined by Proposition~\ref{prop:Poisson})
\begin{equation}
  \label{eq:Poisson_eq_velocity}
  (\partial_t + \cA_0 + \eta \cA_1) \Phi_\eta(t,q,p) = M^{-1}p -  \mathcal{V}_\eta, 
  \qquad
  \qquad \int_\cE \Phi_\eta(t,q,p) \, \psi_\eta(t,q,p) \, dt \, dq \, dp = 0.
\end{equation}
We can then state the following convergence result.

\begin{theorem}
\label{thm:diffusion}
As $\eps \to 0$, the rescaled process $Q_t^{\eta,\eps}$ converges weakly on any finite time interval to the effective $d$-dimensional Brownian motion 
\[
d\overline{Q}_t = \sqrt{2} \, \mathscr{D}_\eta^{1/2} \, dB_t,
\]
with symmetric, positive definite covariance matrix $\mathscr{D}_\eta$ defined by its action on test vectors:
\[
\forall \xi \in \RR^d, \qquad \xi^T \mathscr{D}_\eta \xi = \frac{\gamma}{\beta} \int_{\cE} \left| \nabla_p \left( \xi^T \Phi_\eta \right)\right|^2  \psi_\eta~;
\]
and initial conditions $\overline{Q}_0 \sim \widetilde{\psi}_\eta(q)\,dq$, where
\[
\widetilde{\psi}_\eta(q) = \int_{\RR^d} \psi_\eta(0,q,p) \, dp.
\]
\end{theorem}

The proof of Theorem~\ref{thm:diffusion}, based on the decomposition technique presented in~\cite{KS86} (see also~\cite{KomorowskiLandimOlla2012} for an up-to-date account and further references) is quite standard. It is nonetheless provided in Section~\ref{sec:proof_diffusion} for completeness. Note that a little more work would allow to obtain convergence rates, by slightly extending the approach from~\cite{HP04}. Related homogenization results for the overdamped dynamics with space-time periodic coefficients can be found in~\cite{garnier}.

\begin{remark}
Formal asymptotic expansions for the corresponding Fokker-Planck equation would lead to generally nonsymmetric homogenized diffusion matrix, whose symmetric part is the covariance matrix~$\mathscr{D}_{\eta}$. Indeed, when the microscopic dynamics is reversible, then the homogenized diffusion matrix is symmetric~\cite{Pavliotis2010}. The homogenized diffusion matrix can be symmetric also for nonreversible dynamics, when additional symmetries are present~\cite{thesis}. The antisymmetric part can affect the homogenized dynamics when the homogenized diffusion matrix is space dependent. This would be the case if we considered locally periodic coefficients, as in~\cite{garnier}.
\end{remark}

\begin{remark}[Relationship to the position-dependent diffusion matrix]
 A simple computation (based on~\eqref{eq:equation_to_expand} below) shows that, when $\eta = 0$, 
 the effective diffusion matrix $\mathscr{D}_0$ is related to the position dependent diffusion matrices
 $D_0(q)$ introduced in Proposition~\ref{prop:LR_velocity} as
 \[
 \mathscr{D}_0 = \int_{\mathcal{M}} D_0(q) \, \widetilde{\mu}(q) \, dq.
 \]
 The homogenization problem for the equilibrium dynamics is well studied and the properties of the diffusion tensor $\mathscr{D}_0$ (i.e. its scaling with respect to the friction coefficient) are understood, at least in one dimension~\cite{papan_varadhan,HP08}.
\end{remark}

\subsection{Properties of the diffusion matrix for small forcings}

A careful inspection of the expression of the effective diffusion shows that $\mathscr{D}_\eta = \mathscr{D}_0 + \mathrm{O}(\eta^2)$ if the forcing has average~0 in time for all configurations~$q$; whereas in general a first order response $\mathscr{D}_\eta = \mathscr{D}_0 + \mathrm{O}(\eta)$ arises when the time average of the forcing is not trivial. This is made precise in the proposition below.

\begin{prop}
\label{prop:small_eta_ppties_D}
There exists $r > 0$ such that, for $|\eta| \leq r$ and any direction $\sL \in \RR^d$,  
\[
\sL^T \mathscr{D}_\eta \sL = \sL^T \mathscr{D}_0 \sL + \eta \sL^T \mathpzc{D}_1 \xi + \eta^2 \widetilde{\mathpzc{D}}_{\eta,\xi},
\]
where $\mathscr{D}_0$ is the effective diffusion of the equilibrium dynamics corresponding to $\eta = 0$ and $\widetilde{\mathpzc{D}}_{\eta,\xi}$ is uniformly bounded for $|\eta| \leq r$ and $|\xi| \leq 1$. When the external force has time average~0 for all configurations, namely when
\begin{equation}
\label{eq:null_time_avg_F}
\forall q \in \mathcal{M}, \qquad \int_{T\mathbb{T}} F(t,q) \, dt = 0,
\end{equation}
then the first order correction vanishes: $\mathpzc{D}_1 = 0$.
\end{prop}

\subsubsection*{Numerical illustration}

We now present numerical experiments illustrating the previous theoretical results. We solve numerically the Langevin dynamics~\eqref{eq:Langevin}, discretized with~\eqref{eq:num_Langevin} for the potential~\eqref{eq:pot2d}, so that we expect a coupling between the directions~$x$ and~$y$. We calculate the drift vector and the diffusion matrix using empirical averages of trajectories generated by~\eqref{eq:num_Langevin}. One could also, in principle, solve the stationary Fokker-Planck equation and the Poisson equation using a spectral method, as was done in~\cite{LatPavlKram2013, PavlVog08} for the equilibrium dynamics and for constant external forcings. Here we rely on Monte Carlo simulations. The drift vector and the diffusion matrix are respectively given by
\begin{equation}
 \mathcal{V}_\eta=\lim_{t\to \infty} \frac{\mathbb{E}\Big( \mathscr{Q}^\eta_t-\mathscr{Q}^\eta_0 \Big)}{t}
\end{equation}
and
\begin{equation}
 \mathscr{D}_\eta=\lim_{t\to \infty}\frac{\mathbb{E}\Big( \left[\mathscr{Q}^\eta_t- \mathbb{E}(\mathscr{Q}_t^\eta)\right]\otimes  \left[\mathscr{Q}_t^\eta - \mathbb{E}(\mathscr{Q}_t^\eta)\right]\Big)}{2t}.
\end{equation}
We are interested both in the linear response regime $\eta\to 0$ (for which we expect $\mathscr{D}_\eta\to \mathscr{D}_0$) and the behavior of $\mathscr{D}_\eta$ for large $\eta$. To estimate the drift coefficients and the diffusion tensor, we proceed by a multiple replica strategy. The system is initialized by positions and momenta $(q^0,p^0)$ distributed according to the nonequilibrium measure~$\psi_\eta(0,q,p)\, dq\,dp$. In practice we run the dynamics~\eqref{eq:num_Langevin} for a simulation time $\tau_{\rm neq}$ sufficiently large for the system to converge towards the nonequilibrium steady-state.

For simplicity, we consider nonequilibrium forcings with general form
\[
F(t,q) = F_0(q)\cos(\omega t).
\]
In the example we considered, we picked $F_0$ as given by~\eqref{eq:NM_force} and chose either $\omega=2\pi$ (time dependent forcing) or $\omega=0$ (time-independent forcing). To perform the computations, we fix a (large) number $M$ of replicas and a simulation time $\tau_{\rm sim}$ sufficiently large for the asymptotic diffusive regime to be attained. The drift coefficient and diffusion matrix are evaluated at time $\tau_{\rm sim}$ by empirical averages. More precisely,
\begin{equation}
 \widehat{\mathcal{V}}_\eta^{\tau_{\rm sim},M}=\frac{1}{M\tau_{\rm sim}}\sum_{k=0}^{M-1} \mathscr{Q}^{k,\tau_{\rm sim}/\Delta t},
\end{equation}
where the approximation of the unprojected position of the $k$th replica at time $n\Delta t$ is
\[
\mathscr{Q}^{k,n} = q^{k,0} + \sum_{i=0}^n M^{-1} p^{k,i},
\]
while
\begin{equation}
\widehat{\mathscr{D}}_\eta^{\tau_{\rm sim},M}=\frac{1}{2\tau_{\rm sim}}\left[\frac{1}{M}\sum_{k=0}^{M-1} \left(\mathscr{Q}^{k,\tau_{\rm sim}/\Delta t}\right)^2-\left(\frac{1}{M}\sum_{k=0}^{M-1} \mathscr{Q}^{k,\tau_{\rm sim}/\Delta t}\right)^2\right],
\end{equation}
In the simulations reported below, we used $M=2.5 \times 10^7$, $\tau_{\rm sim} =1500$, $\tau_{\rm neq} = 500$, $\Delta t = 0.01$, $\gamma=1$, $\beta = 1$. Since we use an independent replica strategy, error bars are deduced from the empirical variance. They are not reported here since statistical errors are in all cases below 10\% in relative accuracy.

Figure~\ref{fig:diff} depicts the component of the diffusion matrix as a function of the forcing strength $\eta$. Figure~\ref{fig:spec} depicts the spectrum of the diffusion matrix as a function of $\eta$. As can be seen, the $\mathscr{D}_{\eta,xx}$ component increases much more than the other components as $\eta$ becomes large. The $\mathscr{D}_{\eta,yy}$ component on the other hand remains almost stationary. A zoom on the small $\eta$ variations of $\mathscr{D}_{\eta,xx}$ is presented in Figure~\ref{fig:diff_zoom}, together with a least-square fit of the form
\[
\mathscr{D}_{\eta,xx} = \mathscr{D}_{0,xx} + a \eta + b \eta^2.
\]
The leading coefficient~$a$ is found to be more or less equal to~0 for time-dependent forcings, as predicted by Proposition~\ref{prop:small_eta_ppties_D}. It is however also found to be more or less equal to~0 for the time-independent forcing under consideration.

\begin{figure}[h]
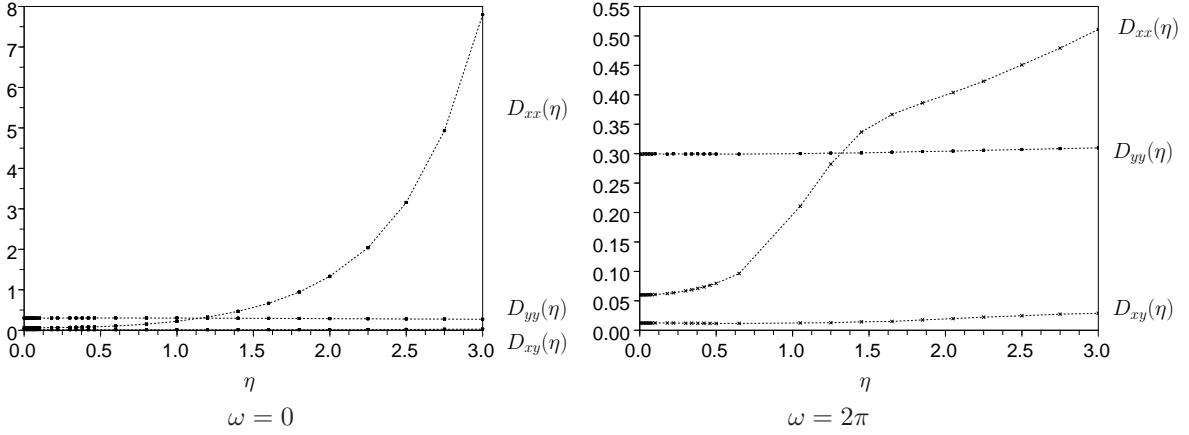

\begin{tabular}{cc}
 \scalebox{0.4}{\input{D_space.pstex_t}}&\qquad \scalebox{0.4}{\input{D_time.pstex_t}}\\
\vspace{0.5cm}
$\omega=0$& $\omega=2\pi$
\end{tabular}
\caption{\label{fig:diff} Components of the diffusion matrix as a function of the forcing for time-independent ($\omega =0$) and space-time dependent ($\omega = 2 \pi$) forcings.}
\end{figure}

\begin{figure}[h]
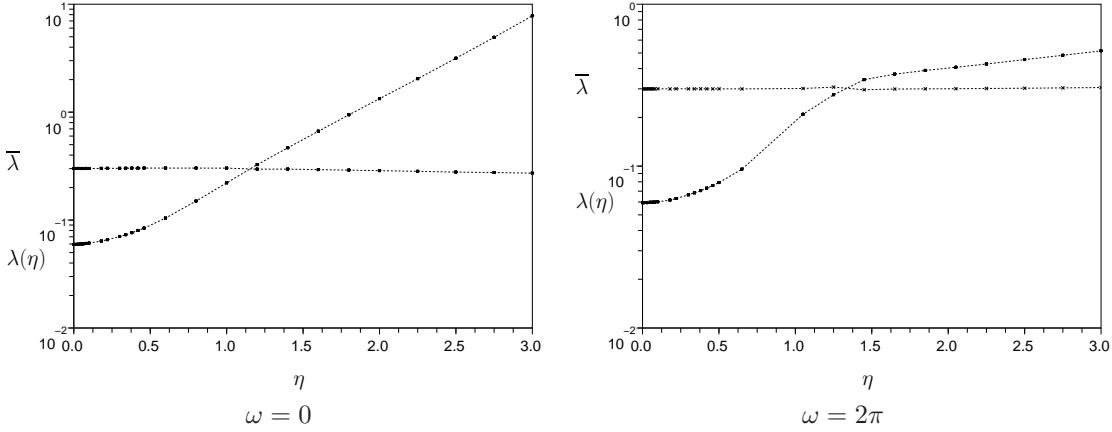

\begin{tabular}{cc}
 \scalebox{0.4}{\input{spec_space.pstex_t}}&\scalebox{0.4}{\input{spec_time.pstex_t}}\\
\vspace{0.5cm}

$\omega=0$& $\omega=2\pi$
\end{tabular}
\caption{\label{fig:spec} Spectrum of the diffusion matrix and as a function of forcing for space and space-time dependent forcing}
\end{figure}

\begin{figure}[h]
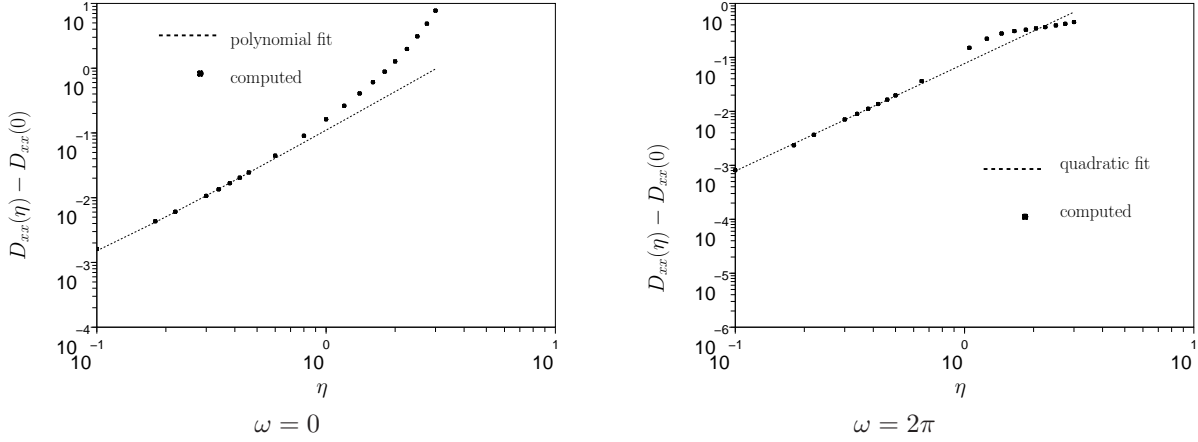

\begin{tabular}{cc}
 \scalebox{0.4}{\input{rep_lin_D_space.pstex_t}}&\qquad \scalebox{0.4}{\input{rep_lin_D_time.pstex_t}}\\
\vspace{0.5cm}
$\omega=0$& $\omega=2\pi$
\end{tabular}
\caption{\label{fig:diff_zoom} Zoom on the small $\eta$ variations of the component $\mathscr{D}_{\eta,xx}- \mathscr{D}_{0,xx}$, together with a quadratic fit $a \eta + b \eta^2$. For $\omega = 2\pi$, we find $a = 0.00016$ and $b = 0.077$, while $a = 0.0034$ and $b = 0.11$ for $\omega = 0$. }
\end{figure}

\section{Proofs of our results}
\label{sec:proofs}

\subsection{Proofs of Propositions~\ref{prop:inv_meas_cv} and~\ref{prop:LLN}}
\label{sec:prop_inv_meas_cv}

The idea of the proof is the following. We first start by establishing the convergence of the distribution of a sampled Markov chain obtained by considering the time-inhomogeneous process~\eqref{eq:Langevin} at times $nT$ (a standard idea, considered in~\cite{HK10} for instance). An invariant measure is then obtained by evolving the invariant measure of the sampled process over a period. 

\paragraph{Exponential convergence of a sampled Markov chain.}
We introduce the Markov chain $(Q^\eta_n,P^\eta_n) = (q^\eta_{nT},p^\eta_{nT})$, where the time-inhomogeneous process $(q^\eta_t,p^\eta_t)$ is started at time $t=0$ from $(Q^\eta_0,P^\eta_0) = (q_0,p_0)$. We have indicated explicitly the dependence on the forcing magnitude~$\eta$ although all the estimates below will hold uniformly with respect to this parameter as long as it remains the bounded interval $[-\eta_*,\eta_*]$. The generator $U_{T,\eta}$ of the Markov chain $(Q_n^\eta,P_n^\eta)$ is defined as
\[
\left(U_{T,\eta}f\right)(q,p) = \mathbb{E}\left(\left. f(Q_{n+1}^\eta,P_{n+1}^\eta) \, \right| \, (Q_n^\eta,P_n^\eta) = (q,p) \right).
\] 
The convergence result for the sampled Markov chain, with rates uniform in~$\eta$, is based on the following two lemmas (proved at the end of this section).

\begin{lemma}[Uniform Lyapunov condition]
\label{lem:Lyapunov}
For any $n \geq 1$ and $\eta_* > 0$, there exist $b >0$ and $a \in [0,1)$ such that, for all $\eta\in [-\eta_*,\eta_*]$,
\begin{equation}
\label{eq:Lyapunov}
U_{T,\eta} \Lin \leq a \Lin + b.
\end{equation}
\end{lemma}

\begin{lemma}[Uniform minorization condition]
\label{lem:minorization}
Fix any $p_{\rm max} > 0$. There exists a probability measure~$\nu$ on~$\mathcal{M} \times \RR^d$ and a constant $\kappa > 0$ such that, for all $\eta\in [-\eta_*,\eta_*]$,
\[
\forall B \in \mathscr{B}(\mathcal{M} \times \RR^d), 
\qquad 
\mathbb{P}\Big( \left(Q^\eta_{k+1},P^\eta_{k+1}\right) \in B \, \Big| \, \left|P^\eta_k\right| \leq p_{\rm max} \Big) \geq \kappa \, \nu(B),
\]
where $\mathscr{B}(X)$ are the Borel sets of $X$.
\end{lemma}

From the results of~\cite{HM11,MeynTweedie}, we can then state the following uniform convergence result for the sampled chain. 

\begin{prop}
\label{prop:cv_sampled_proc}
Fix $\eta^* > 0$ and $n \geq 1$. There exist $\lambda, C > 0$ and probability measures $m_\eta(q,p) \, dq \,dp$ such that, for any $f \in L^\infty_\Lin$ and any $\eta \in [-\eta^*,\eta^*]$,
\begin{equation}
\label{eq:cv_over_one_period}
\left\| U^k_{T,\eta}f - \int_{\M \times \mathbb{R}^d} f(q,p) \, m_\eta(q,p) \, dq \, dp\right\|_{L^\infty_\Lin} \leq C \rme^{-\lambda k T} \| f \|_{L^\infty_\Lin}.
\end{equation}
\end{prop}

Moreover, the integration of the inequality~\eqref{eq:Lyapunov} with respect to $m_\eta$ gives the moment estimate
\[
\int_{\mathcal{M} \times \RR^d} \Lin \, m_\eta \leq \frac{b}{1 - a}.
\]

\paragraph{Law of Large Numbers for the sampled chain.}
To obtain the Law of Large Numbers for all initial conditions, we use the following property (proved at the end of this section). 
 
\begin{lemma}
\label{lem:smoothness_UT}
The generator $U_{T,\eta}$ has a transition kernel which is absolutely continuous with respect to Lebesgue measure and positive. 
\end{lemma}

This property implies that the chain is irreducible with respect to the Lebesgues measure, and also gives the positivity of $m_\eta(q,p)$.  
In view of the Lyapunov condition~\eqref{eq:Lyapunov} and relying on~\cite[Theorem~17.0.1]{MeynTweedie}, we can then conclude that, for any $f \in L^\infty_\Lin$,
\begin{equation}
  \label{eq:LLN_sampled}
  \frac1N \sum_{n=1}^N f(Q^\eta_n,P^\eta_n) \xrightarrow[N\to+\infty]{} \int_{\mathcal{M} \times \RR^d} f \, m_\eta \qquad \mathrm{a.s.}
\end{equation}
for almost all initial conditions~$(Q^\eta_0,P^\eta_0)$. In fact, convergence occurs for all initial conditions and not simply for almost all initial conditions. This is a consequence of the smoothness of the transition probability, ensured by Lemma~\ref{lem:smoothness_UT} (the chain is Harris recurrent, see~\cite[Corollary~1]{Tierney94}, based on~\cite{Nummelin}).

\paragraph{Convergence for the time-inhomogeneous process.}
Convergence results similar to the ones§ stated in Proposition~\ref{prop:cv_sampled_proc} and in~\eqref{eq:LLN_sampled} hold for the sampled chains $(Q^{\eta,\theta}_n,P^{\eta,\theta}_n) = (q^\eta_{nT+\theta},p^\eta_{nT+\theta})$, uniformly in~$\theta \in T \mathbb{T}$ with associated invariant measures~$m_{\eta,\theta}$ (still relying on~\cite{HM11}). In fact, these sampled chains generate the same evolution as the sampled chain defined above up to a time shift in the nonequilibrium forcing, \textit{i.e.} upon replacing $F(t,q)$ by $F(\theta+t,q)$. The exponential convergence~\eqref{eq:cv_exp_time_inhomog} and the Law of Large Numbers~\eqref{eq:LLN} then follow upon defining 
\[
\forall \theta \in T\mathbb{T}, 
\qquad 
\psi_\eta(\theta,q,p) = \frac1T m_{\eta,\theta}(q,p).
\]

\paragraph{Properties of the invariant measure.}
Now that we have proved the exponential convergence to the steady-state, we characterize the invariant measure as the weak solution of an appropriate Fokker-Planck equation. Note that, for a smooth function~$f$,
\[
\frac{d}{ds}\left[ \mathbb{E}\Big(f([s],q^\eta_s,p^\eta_s)\Big) \right] = \mathbb{E}\left[\Big( (\partial_s + \cA_0 + \eta \cA_1)f \Big)([s],q^\eta_s,p^\eta_s)\right].
\]
Passing to the limit $s = nT + \theta$ with $n \to +\infty$ and $\theta \in T\mathbb{T}$, we see that 
\[
\frac{d}{d\theta}\left( \int_{\mathcal{M} \times \RR^d} f(\theta,q,p) \, \psi_\eta(\theta,q,p) \, dq \, dp \right) = \int_{\mathcal{M} \times \RR^d} (\partial_\theta + \cA_0 + \eta \cA_1)f \, \psi_\eta,
\]
which gives~\eqref{eq:FokkerPlanck} in the sense of distributions. The fact that $\psi_\eta$ is smooth and that~\eqref{eq:FokkerPlanck} actually holds pointwise is a consequence of the following lemma.

\begin{lemma}
The operators $\partial_t + \cA_0 + \eta \cA_1$ and $-\partial_t + \cA_0^\dagger + \eta \cA_1^\dagger$ (considered on $L^2(\mathcal{E})$) are hypoelliptic.
\end{lemma}

\begin{proof}
We use H\"ormander's criterion, and present the proof for $\partial_t + \cA_0 + \eta \cA_1$. The proof for the adjoint of this operator is similar. Define $X_0 = \partial_t + M^{-1}p^T (\nabla_q-\gamma \nabla_p) - \nabla V^T\nabla_p + \eta \cA_1$ and 
\[
X_{1,i} = \sqrt{\frac{\gamma}{\beta}} \partial_{p_i}. 
\]
These operators are chosen so that  
\[
\partial_t + \cA_0 + \eta \cA_1 = X_0 + \sum_{i=1}^d X_{1,i}^\dagger X_{1,i},
\]
In addition, denoting by $[L_1,L_2] = L_1L_2 - L_2L_1$ the commutator between two operators $L_1$ and~$L_2$, a simple computation shows that
\[
[X_{i,1},X_0] = \sqrt{\frac{\gamma}{\beta}} \sum_{j = 1}^d \left(M^{-1}\right)_{ij} (\partial_{q_j}-\gamma \partial_{p_j}),
\]
from which derivatives in the directions $q_j$ are recovered. Linear combinations with $X_0$ and $X_{1,i}$ finally allow to recover $\partial_t$, so that the Lie algebra generated by $X_0$, $X_{1,i}$ and $[X_{1,i},X_0]$ is full.
\end{proof}

\paragraph{Proof of the technical results.}
We conclude this section by giving the proofs of the technical lemmas used above.

\begin{proof}[Proof of Lemma~\ref{lem:Lyapunov}]
A simple computation shows that 
\begin{equation}
  \label{eq:evolution_P_one_step}
  P_{k+1}^\eta = \rme^{-\gamma T M^{-1}} P_k^\eta + \mathscr{F}_k + \mathscr{G}_k,
\end{equation}
with
\[
\mathscr{F}_k = \int_{0}^{T} \rme^{-\gamma (T-s) M^{-1}} \Big(-\nabla V(q^\eta_{kT+s}) + \eta F(s,q^\eta_{kT+s})\Big) \, ds
\]
and
\[
\mathscr{G}_k = \sqrt{\frac{2\gamma}{\beta}}\int_{0}^{T} \rme^{-\gamma (T-s) M^{-1}} dW_{kT+s}.
\]
Note that the martingale $\mathscr{G}_k$ is distributed according to a Gaussian with mean~0 and covariance $M(1-\exp(-2\gamma T M^{-1}))/\beta$, hence has finite moments of all orders; while a simple uniform bound on $\mathscr{F}_k$ is for instance 
\begin{equation}
  \label{eq:bound_drift_one_step}
  |\mathscr{F}_k| \leq \mathfrak{F}_\eta = T( \| \nabla V\|_{L^\infty(\mathcal{M})} + |\eta| \| F \|_{L^\infty(T\mathbb{T} \times \mathcal{M})} )\leq \mathfrak{F}_{\eta_*}.
\end{equation}
Note first that
\[
\Lin(Q_{k+1}^\eta,P_{k+1}^\eta) = 1 + \left( \left|\rme^{-\gamma T M^{-1}} P_k^\eta + \mathscr{F}_k + \mathscr{G}_k\right|^2 \right)^n. 
\]
Now, introducing $0 < \alpha = \rme^{-\gamma T/m} < 1$ where 
\begin{equation}
  \label{eq:bound_mass_matrix}
  \frac1m \, \mathrm{Id} \leq M \leq m \, \mathrm{Id},
\end{equation}
it holds
\begin{align}
\left|\rme^{-\gamma T M^{-1}} P_k^\eta + \mathscr{F}_k + \mathscr{G}_k\right|^2 & \leq \alpha^2 \left|P_k^\eta\right|^2 + \alpha \mathfrak{F}_\eta \left|P_k^\eta\right| + \mathfrak{F}_\eta^2 + \mathscr{G}_k^2 + 2 \mathscr{G}_k^T \left( \rme^{-\gamma T M^{-1}} P_k^\eta + \mathscr{F}_k\right) \nonumber \\ 
& \leq \alpha^2 (1 + \eps) \left|P_k^\eta\right|^2 + \left(2+\frac{1}{4\eps}\right)\mathfrak{F}_\eta^2 + 2\left|\mathscr{G}_k\right|^2 + 2 \mathscr{G}_k^T \rme^{-\gamma T M^{-1}} P_k^\eta, \label{eq:bound_n2}
\end{align}
for any $\eps > 0$, chosen so that $\alpha^2(1+\eps)<1$. Denoting by $\mathcal{F}_k$ the filtration induced by the Markov chain up to the $k$th step, it holds
\[
\mathbb{E}\left[ \mathscr{G}_k^T P_k^\eta \Big| \, \mathcal{F}_k\right] = 0,
\]
so that
\[
\mathbb{E}\left[\left. \Li_2(Q_{k+1}^\eta,P_{k+1}^\eta) \right| \, \mathcal{F}_k\right] \leq \alpha^2 (1+\eps) \Li_2(Q_{k}^\eta,P_{k}^\eta) + C_\eps
\]
for some constant $C_\eps > 0$ which can be chosen to be independent on $\eta \in [\eta_*,\eta_*]$ (but depending on~$\eta_*$ of course). This gives the result in the case $n=2$.

For a general index $n \geq 2$, we take the $n$th power of the bound~\eqref{eq:bound_n2}. The leading term is $\alpha^{2n}(1+\eps)^n$, which is still strictly smaller than~1 with the above choice of~$\eps$. Terms with odd powers of $\mathscr{G}_k$ vanish when taking the expectation. The expectation of the remainder is a linear combination of terms $(P_k^\eta)^{2k}$ with $k \leq n-1$. Therefore,
\[
\mathbb{E}\left[\left. \Lin(Q_{k+1}^\eta,P_{k+1}^\eta) \right| \, \mathcal{F}_k\right] \leq \alpha^{2n}(1+\eps)^n|P_k^\eta|^{2n} + \mathcal{P}(|P_k^\eta|),
\]
where $\mathcal{P}$ is a polynomial of degree at most~$2(n-1)$, with bounded, positive coefficients. As $|P_k^\eta| \to +\infty$, 
\[
\frac{\mathcal{P}(|P_k^\eta|)}{|P_k^\eta|^{2n}} \to 0.
\]
This shows that, for any $\delta > 0$, there exists $C_\delta > 0$ such that $\mathcal{P}(|P_k^\eta|) \leq \delta |P_k^\eta|^{2n} + C_\delta$ (consider a radius $R_\delta > 0$ for which $\mathcal{P}(r)/r^{2n} \leq \delta$ when $r \in [0,R_\delta]$ and define $C_\delta$ as the supremum of~$\mathcal{P}$ over~$[0,R_\delta]$). The conclusion easily follows upon choosing $\delta > 0$ sufficiently small.
\end{proof}

\begin{proof}[Proof of Lemma~\ref{lem:minorization}]
We use the same strategy as in the proof of~\cite[Lemma~5]{LMS13}. Note first that it is sufficient to prove the result for Borel sets~$B = B_q \times B_p \subset \mathcal{M}$ where $B_q \subset \mathcal{M}$ while $B_p \subset \mathbb{R}^{dN}$. The idea is to relate the dynamics~\eqref{eq:Langevin} to a Langevin dynamics without forces (\textit{i.e.} $\nabla V + \eta F = 0$) for which it is not difficult to construct a minorizing measure. More precisely, a simple computation based on an equality similar to~\eqref{eq:evolution_P_one_step} shows that
\begin{equation}
  \label{eq:evolution_Q_one_step}
  Q^\eta_{k+1} = Q^\eta_k + \frac{1-\alpha_T}{\gamma} P^\eta_k + \widetilde{\mathscr{F}}_k + \widetilde{\mathscr{G}}_k,
\end{equation}
with $\alpha_T = \exp(-\gamma T M^{-1})$, 
\[
\widetilde{\mathscr{G}}_k = \sqrt{\frac{2\gamma}{\beta}} \int_0^T \int_{0}^{s} \rme^{-\gamma (s-r) M^{-1}} M^{-1} dW_{kT+r} \, ds,
\]
and where
\[
\widetilde{\mathscr{F}}_k = \int_0^T \int_{0}^{s} \rme^{-\gamma (s-r) M^{-1}} M^{-1}\Big(-\nabla V(q^\eta_{kT+r}) + \eta F(r,q^\eta_{kT+r}\Big) \, dr \, ds
\]
is uniformly bounded by $m T \mathfrak{F}_\eta$ (where $\mathfrak{F}_\eta$ and $m$ are defined in~\eqref{eq:bound_drift_one_step} and~\eqref{eq:bound_mass_matrix}, respectively). The random variables $\widetilde{\mathscr{G}}_k,\mathscr{G}_k$ are correlated centered Gaussian random variables, independent of $\widetilde{\mathscr{G}}_l,\mathscr{G}_l$ for $l \neq k$. Their covariance reads 
\[
\mathpzc{V} = \mathbb{E}\left[ \left(\widetilde{\mathcal{G}}_k,\mathcal{G}_k\right)^T \left(\widetilde{\mathcal{G}}_k,\mathcal{G}_k\right) \right] = \begin{pmatrix} \mathpzc{V}_{qq} & \mathpzc{V}_{qp} \\ \mathpzc{V}_{qp} & \mathpzc{V}_{pp} \end{pmatrix},
\]
with
\[
\left\{ 
\begin{aligned}
\mathpzc{V}_{qq} & = \frac{1}{\beta \gamma}\left(2T - \frac{M}{\gamma}\left(3 - 4 \, \alpha_T + \alpha_T^2\right)\right),\\
\mathpzc{V}_{qp} & = \frac{M}{\beta \gamma}\left(1-\alpha_T\right)^2, \\
\mathpzc{V}_{pp} & = \frac{M}{\beta} \left(1-\alpha_T^2\right).
\end{aligned}
\right.
\]
Therefore, in view of~\eqref{eq:evolution_P_one_step} and~\eqref{eq:evolution_Q_one_step},
\[
\mathbb{P}\Big( \left(Q^\eta_{k+1},P^\eta_{k+1}\right) \in B \, \Big| \, \left|P^\eta_k\right| \leq p_{\rm max} \Big) = \mathbb{P}\Big( \left(\widetilde{\mathscr{G}}_k,\mathscr{G}_k\right) \in (B_q - \mathscr{Q}_k) \times (B_p - \mathscr{P}_k) \, \Big| \, \left|P^\eta_k\right| \leq p_{\rm max} \Big),
\]
where $\mathscr{Q}_k = Q_k^\eta + (1-\alpha_T) P^\eta_k/\gamma + \widetilde{\mathscr{F}}_k$ and $\mathscr{P}_k = \alpha_T P_k^\eta + \mathscr{F}_k$ are uniformly bounded in norm by a constant~$R>0$ which depends only on~$p_{\rm max}$ and $\eta_*$. The proof is therefore concluded by defining the probability measure
\[
\nu(B_q \times B_p) = Z_R^{-1} \inf_{|\mathscr{Q}|, |\mathscr{P}| \leq R} \int_{(B_q - \mathscr{Q}) \times (B_p-\mathscr{P})} \exp\left( -\frac{x^T \mathpzc{V}^{-1} x}{2} \right) \, dx,
\]
with $Z_R > 0$ and $\kappa  = (2\pi)^{-dN} \mathrm{det}\left(\mathpzc{V}\right)^{-1/2} Z_R$.
\end{proof}

\begin{proof}[Proof of Lemma~\ref{lem:smoothness_UT}]
To prove the smoothness of the transition kernel, we use the same argument as in the proof of~\cite[Lemma~2.2]{Talay02} and rewrite, thanks to Girsanov's formula, the nonequilibrium evolution $(q_t^\eta,p_t^\eta)$ solution of~\eqref{eq:Langevin} in terms of the equilibrium evolution~$(q_t^0,p_t^0)$ given by~\eqref{eq:Langevin_eq}. Starting from an initial condition~$(q_0,p_0)$ at time~$t=0$ and integrating over one period~$T$, the Girsanov transform (see \textit{e.g.}~\cite{KS91}) shows that the law of $(q_t^\eta,p_t^\eta)$ is absolutely continuous with respect to the law of~$(q_t^0,p_t^0)$ and that, for any bounded, measurable function~$f$,  
\[
\mathbb{E}_{(q_0,p_0)}\left[ f\left(q_T^\eta,p_T^\eta\right)\right] = \mathbb{E}_{(q_0,p_0)}\left[ f\left(q_T^0,p_T^0\right) \, Z_T^0\right], 
\]
where
\[
Z_T^0 = \exp\left(\eta \sqrt{\frac{\beta}{2\gamma}} \int_0^T F(t,q_t^0) \cdot dW_t - \frac{\eta^2 \beta}{4\gamma} \int_0^T |F(t,q_t^0)|^2 dt \right).
\]
Now, standard hypoellipticity results show that the random variable~$(q_T^0,p_T^0)$ has a smooth density with respect to the Lebesgue measure~$dq\,dp$ on~$\mathcal{M} \times \RR^d$ (see for instance the discussion in~\cite[Section~7]{rey-bellet}). We conclude that~$(q_T^\eta,p_T^\eta)$ also has a density density with respect to the Lebesgue measure.

The positivity of the transition kernel can be proved using a standard control argument, see for instance~\cite[Section~6]{rey-bellet} or~\cite{MSH02,Talay02}.
\end{proof}

\subsection{Proof of Proposition~\ref{prop:inv_meas_charact}}
\label{sec:inv_meas_proof_charact}

The idea of the proof is to show, by a perturbative construction, the existence of a solution to the extended Fokker-Planck equation~\eqref{eq:FokkerPlanck}. Proposition~\ref{prop:inv_meas_cv} then gives the uniqueness. We follow the strategy summarized in~\cite{HDR} and already used in~\cite{JS12}. See also~\cite{KomOlla2005}.

\paragraph{Series expansions of the invariant measure.}
Instead of the flat space $L^2(\mathcal{E})$, we consider as a reference space the Hilbert space $\cH = L^2(\mathcal{E},\mu) \cap \{ \mathbf{1} \}^\perp$, the space of functions $f(t,q,p)$ such that 
\[
\int_\mathcal{E} |f(t,q,p)|^2 \mu(q,p) \, dq\,dp\,dt < +\infty, 
\qquad 
\int_\mathcal{E} f(t,q,p) \, \mu(q,p) \, dq\,dp\,dt = 0.
\]
The Fokker-Planck equation~\eqref{eq:FokkerPlanck} can be rewritten as
\[
(-\partial_t+\cA_0^* + \eta \cA_1^*)\rho_\eta = 0, 
\qquad 
\psi_\eta = \rho_\eta \mu,
\qquad 
\int_\cE \rho_\eta \mu = 1,
\]
where adjoints are now taken on~$\cH$. Formally,
\begin{equation}
\label{eq:formal_expansion_varrho}
\rho_\eta = \Big(\mathrm{Id} + \eta (-\partial_t+\cA_0^*)^{-1} \cA_1^* \Big)^{-1} \mathbf{1} = \sum_{n=0}^{+\infty} (-\eta)^n \left[ (-\partial_t+\cA_0^*)^{-1} \cA_1^* \right]^n \mathbf{1}
\end{equation}
To make this argument rigorous, we use two technical results (proved at the end of this section).

\begin{lemma}
  \label{lem:pt_A0_bounded}
  The operator $\partial_t + \cA_0$ is invertible on~$\cH$.
\end{lemma}

\begin{lemma}
  \label{lem:A1_T0_bounded}
  The operator $\cA_1$ is $(\partial_t+\cA_0)$-bounded on~$L^2(\mathcal{E},\mu)$: There exists $a,b > 0$ such that, for all smooth functions~$f$,
  \[
  \| \cA_1 f \|_{L^2(\mathcal{E},\mu)} \leq a \| (\partial_t+\cA_0) f \|_{L^2(\mathcal{E},\mu)} + b \| f \|_{L^2(\mathcal{E},\mu)}.
  \]
\end{lemma}

We conclude that $\cA_1 (\partial_t + \cA_0)^{-1}$ is a bounded operator from~$\cH$ to~$L^2(\mathcal{E},\mu)$. Its adjoint $(-\partial_t + \cA_0^*)^{-1} \cA_1^*$ is therefore bounded on~$L^2(\mathcal{E},\mu)$. Moreover, a simple computation shows that $\mathrm{Ran}(\cA_1^*) \subset \cH$: indeed,
\[
\int_\mathcal{E} (\cA_1^* f) \mu = \int_\mathcal{E} f (\cA_1 \mathbf{1}) \mu = 0.
\] 
Therefore, the operator $(-\partial_t + \cA_0^*)^{-1} \cA_1^*$ is bounded on~$\cH$. This proves the validity of~\eqref{eq:formal_expansion_varrho} for $|\eta|$ smaller than the spectral radius of $(-\partial_t + \cA_0^*)^{-1} \cA_1^*$.

\paragraph{Proofs of technical results.}
We conclude this section with the proofs of Lemmas~\ref{lem:pt_A0_bounded} and~\ref{lem:A1_T0_bounded}. We first recall the following result, which is a simple consequence of the results in~\cite{EH03,HN04}.

\begin{lemma}
  \label{lem:unif_resolvent_L2}
  The operators $\ri \nu + \cA_0$ are invertible on $L^2(\mathcal{M}\times\RR^d,\mu)$ for any $\nu \neq 0$, with inverses uniformly bounded away from~0: For any $\nu_* > 0$, there exists $C > 0$ such that
  \[
  \sup_{|\nu| \geq \nu_*} \left\| (\ri \nu + \cA_0)^{-1} \right\|_{\mathcal{B}(L^2(\mu))} \leq C,
  \]
  where $\mathcal{B}(L^2(\mu))$ is the Banach space of bounded linear operators on~$L^2(\mu)$. In addition, $\cA_0$ is invertible on $L^2(\mathcal{M}\times\RR^d,\mu) \backslash \mathrm{Ker}(\cA_0)$ with $\mathrm{Ker}(\cA_0) = \mathrm{Span}(\mathbf{1})$.
\end{lemma}

\begin{proof}[Proof of Lemma~\ref{lem:pt_A0_bounded}]
Consider, for a given function $g \in \cH$, the equation $(\partial_t + \cA_0)f = g$. We decompose $f$ and $g$ in Fourier series as 
\[
f = \sum_{n \in \mathbb{Z}} f_n e_n, \qquad g = \sum_{n \in \mathbb{Z}} g_n e_n
\]
with $e_n(t) = \rme^{i n\omega t}$, $g_n \in L^2(\mathcal{M}\times\RR^d,\mu)$ for all $n \in \mathbb{Z}$ and
\[
\int_{\mathcal{M}\times\RR^d} g_0 \, \mu = 0.
\]
Therefore, $(\partial_t + \cA_0)f = g$ can be rewritten as
\[
(\ri n \omega + \cA_0) f_n = g_n.
\]
When $n \neq 0$, we use Lemma~\ref{lem:unif_resolvent_L2} to obtain $f_n = (\ri n \omega + \cA_0)^{-1} g_n$, while $f_0 = \cA_0^{-1} g_0$ is a well defined element of $L^2(\mathcal{M}\times\RR^d,\mu) \cap \{ \mathbf{1} \}^\perp$ since $\cA_0$ is invertible on $L^2(\mathcal{M}\times\RR^d,\mu) \cap \{ \mathbf{1} \}^\perp$. Moreover, there exists a constant $C > 0$ such that $\| f_n \|_{L^2(\mathcal{M}\times\RR^d,\mu)} \leq C \| g_n \|_{L^2(\mathcal{M}\times\RR^d,\mu)}$, from which we deduce that $\| f \|_{L^2(\mathcal{E},\mu)} \leq C \| g \|_{L^2(\mathcal{E},\mu)}$. This proves that $\partial_t+\cA_0$ has a bounded inverse on $\cH$.
\end{proof}

\begin{proof}[Proof of Lemma~\ref{lem:A1_T0_bounded}]
For a smooth function $f$ defined on~$\mathcal{E}$, using the fact that the force $F$ is uniformly bounded (since $F(t,q)$ is smooth and its arguments belong the bounded set $T\mathbb{T} \times \mathcal{M}$), 
\[
\begin{aligned}
\| \cA_1 f\|^2_{L^2(\mathcal{E},\mu)} & \leq \| F \|^2_{L^\infty(T\mathbb{T} \times \mathcal{M})} \| \nabla_p f\|_{L^2(\mathcal{E},\mu)}^2 = -\| F \|^2_{L^\infty(T\mathbb{T} \times \mathcal{M})} \left\langle (\partial_t + \cA_0)f, f\right\rangle_{L^2(\mathcal{E},\mu)} \\
& \leq \| F \|^2_{L^\infty(T\mathbb{T} \times \mathcal{M})} \| f\|_{L^2(\mathcal{E},\mu)} \| (\partial_t + \cA_0) f\|_{L^2(\mathcal{E},\mu)},
\end{aligned}
\]
from which the conclusion easily follows.
\end{proof}

\subsection{Proof of Proposition~\ref{prop:LR_velocity}}
\label{sec:proof_LR_velocity}

We first make precise the equation satisfied by the linear response term $\varrho_1$. From~\eqref{eq:formal_expansion_varrho},
\begin{equation}
  \label{eq:eq_varrho_1}
  (-\partial_t+\cA_0^*)\varrho_1 = -\cA_1^* \mathbf{1} = -\beta p^T M^{-1} F,  
  \qquad 
  \int_\mathcal{E} \varrho_1 \, \mu = 0.
\end{equation}
Decomposing both sides of the equation in Fourier series, we obtain
\begin{equation}
  \label{eq:1OPT}
  (-\ri n \omega + \cA_0^*) \varrho_{1,n} = (\ri n \omega + \cA_0)^* \varrho_{1,n} = -\beta p^T M^{-1} F_n, 
  \qquad
  \int_\mathcal{E} \varrho_{1,n} \, \mu = 0.
\end{equation}
This equation is well posed for $n \neq 0$ in view of Lemma~\ref{lem:unif_resolvent_L2}. For $n=0$, it is also well posed since $p^TM^{-1}F_0$ has a vanishing average with respect to~$\mu(q,p)\,dq\,dp$.

Using the expression~\eqref{eq:1OPT} of the first order correction in~$\eta$ of the invariant measure, the linear response of the average velocity is 
\[
\begin{aligned}
\mathscr{V}(t) & = \lim_{\eta \to 0} \frac{v^\eta(t)}{\eta} \\
& = \sum_{n \in \bZ} e_n(t) \int_\M \int_{\mathbb{R}^d}
\left(M^{-1}p\right) \varrho_{1,n}(q,p) \, \mu(q,p) \, dq \, dp \\
& = -\beta \sum_{n \in \bZ} e_n(t) \int_\M \int_{\mathbb{R}^d}
\left(M^{-1}p\right) \left[ \left(-\ri n\omega + \cA_0^*\right)^{-1} (p^T M^{-1} F_n)\right] \mu(q,p) \, dq \, dp \\
& = -\beta \sum_{n \in \bZ} e_n(t) \int_\M \int_{\mathbb{R}^d}
\left[ \left(\ri n\omega + \cA_0\right)^{-1}\left(M^{-1}p\right)\right] \left(p^T M^{-1} F_{n}(q)\right) \, \mu(q,p) \, dq \, dp.
\end{aligned}
\]

At this stage, we would like to rewrite the inverse operators as time integrals of the semigroup, in order to introduce equilibrium correlations. This is possible in the $H^1(\mu)$ norm for instance, thanks to hypocoercivity estimates (see~\cite{Villani} for background information on the theory of hypocoercivity). In particular, it is proved in~\cite{HP08} that there exist $C,\lambda > 0$ such that $\| \mathrm{e}^{s \cA_0}\|_{\mathcal{B}(\widetilde{H}^1(\mu))} \leq C \rme^{-\lambda s}$ where 
\[
\widetilde{H}^1(\mu) = \left\{ f \in H^1(\mu) \left| \int_{\mathcal{M} \times \RR^d} f(q,p) \, \mu(q,p) \, dq\,dp = 0 \right.\right\}.
\]
This proves the following lemma.

\begin{lemma}
\label{lem:indep_omega}
The following equality as operators on~$\widetilde{H}^1(\mu)$ holds for any $\nu \in \RR$:
\[
(\ri\nu-\cA_0)^{-1} = \int_0^{+\infty} \mathrm{e}^{-\ri \nu s} \mathrm{e}^{s \cA_0} \, ds.
\]
Moreover, the operators $(\ri\nu-\cA_0)^{-1}$ are uniformly bounded on~$\widetilde{H}^1(\mu)$, with a bound independent of $\nu \in \mathbb{R}$.
\end{lemma}

Using this result, and since $M^{-1}p$ has all his components in $\widetilde{H}^1(\mu)$, we rewrite the linear response of the velocity as
\[
\begin{aligned}
\mathscr{V}(t) & = \beta \sum_{n \in \bZ} e_n(t) \int_0^{+\infty} {\rm e}^{\ri n \omega s} \left(\int_\M \int_{\mathbb{R}^d}
\left[ {\rm e}^{s \mathcal{A}_0} \left(M^{-1}p\right) \right] \left(p^TM^{-1}F_{n}(q)\right) \mu(q,p) \, dq \, dp \right) ds \\
& = \beta \sum_{n \in \bZ}  e_n(t) \int_0^{+\infty} {\rm e}^{\ri n \omega s} \, \mathbb{E}\left[ \Big( \left(M^{-1}p_s\right) \otimes \left(M^{-1}p_0\right)\Big) F_{n}(q_0) \right] \, ds.
\end{aligned}
\]
This gives~\eqref{eq:full_LR}.

\subsection{Proof of Proposition~\ref{prop:genuine_spatial_dep}}
\label{sec:proof_genuine_spatial_dep}

Recall first that, by hypoellipticity, $\Phi_0$ is smooth. Assume, by contradiction, that $\Phi_0$ does not depend on~$q$ so that $\Phi_0 = \Phi_0(p) = (\Phi_{0,1}(p),\dots,\Phi_{0,d}(p))$. In this case, 
\begin{equation}
\label{eq:contradiction_Phi_0}
\cA_0 \Phi_{0,i} = -\nabla V \cdot \nabla_p \Phi_{0,i} + \gamma \left( -M^{-1} p \cdot \nabla_p \Phi_{0,i} + \frac1\beta \Delta_p \Phi_{0,i} \right).
\end{equation}
Since 
\begin{equation}
  \label{eq:cA_0_Phi_0,i}
  \cA_0 \Phi_{0,i} = \sum_{j=1}^d \left(M^{-1}\right)_{i,j} p_j
\end{equation}
by definition of~$\Phi_0$, the right-hand side of~\eqref{eq:contradiction_Phi_0} does not depend on~$q$. This means that $\nabla V(q) \cdot \nabla \Phi_{i,0}(p) = 0$ for all $(q,p)$. In view of the non-degeneracy condition~\eqref{eq:non_degeneracy}, this means that $\nabla \Phi_{i,0} = 0$ for all $p \in \mathbb{R}^d$, which is impossible since $\Phi_{0,i}$ cannot be constant in view of~\eqref{eq:cA_0_Phi_0,i}.

\subsection{Proof of Proposition~\ref{prop:SR}}
\label{sec:SR_proof}

The proof is an immediate consequence of the following lemma (upon replacing $\psi$ by the components $p_i$ of~$p$, and changing the sign of~$\omega$). 

\begin{lemma}
\label{lem:SR}
Consider $\psi \in L^2(\mu)$ such that $\cA_0^m \psi \in L^2(\mu)$ for any $0 \leq m \leq n$. Introduce, for a frequency $\omega \neq 0$, the unique solution $\varphi_\omega$ (well defined by Lemma~\ref{lem:unif_resolvent_L2}) of 
\[
(\ri \omega - \cA_0) \varphi_\omega = \psi.
\]
Then, there exist a constant $C_n>0$ and functions $\phi_2,\dots,\phi_{n-1}$ such that, for all $\omega \geq 1$,
\[
\left\| \varphi_\omega - \left(\frac{\psi}{\ri\omega} + \frac{\phi_2}{\omega^2} + \dots + \frac{\phi_{n-1}}{\omega^{n-1}} \right)\right\|_{L^2(\mu)} \leq \frac{C_n}{\omega^n}.
\]
\end{lemma}

\begin{proof}
A simple computation shows that
\[
(\ri \omega - \cA_0) \left(\varphi_\omega - \frac{\psi}{\ri\omega} - \frac{\cA_0 \psi}{(\ri\omega)^2}-\dots-\frac{\cA_0^{n-1}\psi}{(\ri\omega)^n}\right)
= \frac{1}{(\ri \omega)^n} \cA_0^n \psi,
\]
so that, in view of Lemma~\ref{lem:unif_resolvent_L2},
\[
\left\| \varphi_\omega - \frac{\psi}{\ri\omega} - \frac{\cA_0 \psi}{(\ri\omega)^2}-\dots-\frac{\cA_0^{n-1}\psi}{(\ri\omega)^n} \right\|_{L^2(\mu)} 
\leq \frac{1}{|\omega|^n} \| (\ri \omega - \cA_0)^{-1} \|_{\mathcal{B}(L^2(\mu))}
\left\| \cA_0^n \psi \right\|_{L^2(\mu)}.
\]
The proof is concluded by setting 
\[
C_n = \left(\sup_{|\omega| \geq 1} \| (\ri \omega - \cA_0)^{-1} \|_{\mathcal{B}(L^2(\mu))}\right) \left\| \cA_0^n \psi \right\|_{L^2(\mu)},
\]
and $\phi_m = (-\ri)^m \cA_0^{m-1} \psi$ for $m = 2,\dots,n-1$.
\end{proof}

\subsection{Proof of Proposition~\ref{prop:Poisson}}
\label{sec:proof_Poisson}

The first remark is that, by linearity, $\Phi_\eta$ is the sum of the solution of~\eqref{eq:Poisson_equation} with right-hand side $f-\overline{f}_\eta(t)$ (see~\eqref{eq:local_time_average} for the definition of $\overline{f}_\eta(t)$) and 
\[
\widetilde{\Phi}_\eta(t) = \int_0^t \left( \overline{f}_\eta(s) - \int_\cE f \right)ds, 
\]
which corresponds to the solution of~\eqref{eq:Poisson_equation} with right-hand side $
\overline{f}_\eta(t) - \int_\cE f $ (note that this function is, of course, mean zero over $T\mathbb{T}$). In addition, $\widetilde{\Phi}_\eta$ is bounded, and its spatial derivatives vanish. Without restriction of generality, we can therefore restrict our attention to the Poisson equation
\begin{equation}
  \label{eq:Poisson_modified}
  (\partial_t + \cA_0 + \eta \cA_1) \Phi_{\eta}(t,q,p) = f(t,q,p) - \overline{f}_\eta(t), 
  \qquad
  \qquad \int_\cE \Phi_\eta \, \psi_\eta = 0,
\end{equation}
Estimates on solutions of the above equation are obtained by generalizing a technique first employed by Talay in~\cite{Talay02}, and recently carefully rewritten and extended in~\cite[Appendix~A]{Kopec13}. We briefly recall the beginning of the argument here in order to show that it can be extended to processes with time-periodic drivings. The crucial estimate is the one provided by Lemma~\ref{lem:poly_decay} below.

Consider $u_0(t,q,p) = f(t,q,p) - \overline{f}_\eta(t)$, and define $u(s) = \rme^{s(\partial_t + \cA_0 + \eta \cA_1)}u_0$, or equivalently 
\[
u(s;t_0,q_0,p_0)= \mathbb{E}_{t_0,q_0,p_0}\Big(u_0(t_s,\wq^\eta_s,\wp^\eta_s)\Big),
\]
where the process $(t_s,\wq_s,\wp_s)$ has spatial initial conditions $(\wq^\eta_0,\wp^\eta_0) = (q_0,p_0)$ and evolves according to the \emph{time homogeneous} Markovian dynamics with generator 
\begin{equation}
\label{eq:def_T_eta}
T_\eta = \partial_t + \cA_0 + \eta \cA_1, 
\end{equation}
namely
\begin{equation}
\label{eq:Langevin_extended}
\left\{ 
\begin{aligned}
dt_s & = ds, \\
d\wq_s^\eta & = M^{-1} \wp_s^\eta \, ds, \\
d\wp_s^\eta & = \Big( -\nabla V(\wq_s^\eta) + \eta F(t_s,\wq_s^\eta) \Big) ds - \gamma M^{-1} \wp_s^\eta \, ds + \sqrt{\frac{2\gamma}{\beta}} \, dW_s.
\end{aligned}
\right.
\end{equation} 
The solutions of~\eqref{eq:Langevin_extended} for a given initial condition $(t_0,q_0,p_0)$ can be seen as the solutions of~\eqref{eq:Langevin} started at time $t_0$ rather than~0. The exponential convergences~\eqref{eq:cv_exp_time_inhomog} and~\eqref{eq:cv_over_one_period} (and their immediate generalizations for $t_0 \neq 0$) imply that, for any $n \geq 1$ and all initial conditions $(t_0,q_0,p_0)$,
\begin{equation}
  \label{eq:decay_us}
  \left| u(s;t_0,q_0,p_0) \right| \leq C_n \rme^{-\lambda_n s} \| f \|_{L^\infty(L^\infty_{\Li_n})} \Li_n(q,p).
\end{equation}
An integration of this inequality from $s=0$ to $+\infty$ shows that the function
\begin{equation}
  \label{eq:def_Phi_eta_semigroup}
  \Phi_\eta(t,q,p) = -\int_0^{+\infty} u(s;t,q,p) \, ds
\end{equation}
is a well defined, admissible solution of the Poisson equation~\eqref{eq:Poisson_modified}. In fact it is also the only one. This already gives, for some constant $C > 0$ depending only on~$n$ and $\eta_*$, the following pointwise estimate:
\begin{equation}
  \label{eq:Linf_estimate_Phi}
  \| \Phi_\eta \|_{L^\infty(L^\infty_{\Li_n})} \leq C \| f \|_{L^\infty(L^\infty_{\Li_n})}. 
\end{equation}
Upon choosing $l > 2n+d/2$, we also obtain the following weighted $L^2$ exponential convergence from~\eqref{eq:decay_us}:
\begin{equation}
  \label{eq:weighted_L2_decay}
  \int_\mathcal{E} |u(s)|^2 \, \Pi_l \leq \widetilde{C}_{n,l} \rme^{-2\lambda_n s} \| f \|_{L^\infty(L^\infty_{\Li_n})}^2,
\end{equation}
where we have introduced the polynomial weight
\[
\Pi_l(t,q,p) = \left(\frac{1}{1+|p|^2}\right)^l.
\]
The idea of Talay in~\cite[Section~3.4]{Talay02} is to use appropriate mixed derivatives to obtain pointwise decay estimates for the derivatives of~$\Phi_\eta$, by first obtaining decay estimates in weighted $L^2$ spaces and then concluding by appropriate Sobolev embeddings. A key estimate to perform the computation is the following lemma (compare~\cite[Lemma~A.6]{Kopec13}).

\begin{lemma}
  \label{lem:poly_decay}
  For a given linear operator~$L$ and $k \in \mathbb{N}$ (sufficiently large for all the integrals below to be well defined), there exists $A_{k,\eta_*} \geq 0$ such that, for any $\alpha > 0$ and any $\eta \in [-\eta_*,\eta_*]$, 
  \begin{align}
  & \rme^{\alpha s} \int_\cE |Lu(s)|^2 \, \Pi_k + \frac{2\gamma}{\beta} \int_0^s \rme^{\alpha r} \left( \int_\cE \left|\nabla_p Lu(r)\right|^2 \, \Pi_k\right)dr \nonumber \\
    & \ \leq \int_\cE |Lu(0)|^2 \, \Pi_k + (A_{k,\eta_*} + \alpha) \int_0^s \rme^{\alpha r} \left( \int_\cE \left|Lu(r)\right|^2 \, \Pi_k\right)dr + \int_0^s \rme^{\alpha r} \left( \int_\cE Lu(s) \, [L,T_\eta]u(r) \, \Pi_k\right)dr. 
    \label{eq:exp_decay_bound}
  \end{align}
\end{lemma}

\begin{proof}
  The proof is based on the following computations:
  \[
  \begin{aligned}
  \frac{d}{ds}\left(\frac12 \int_\cE |Lu(s)|^2 \, \Pi_k\right) & = \int_\cE Lu(s) \, LT_\eta u(s) \, \Pi_k \\
  & = \int_\cE Lu(s) \, T_\eta Lu(s) \, \Pi_k + \int_\cE Lu(s) \, [L,T_\eta] u(s) \, \Pi_k  \\
  & = \frac12 \int_\cE T_\eta\left( |Lu(s)|^2 \right) \, \Pi_k - \frac{\gamma}{\beta} \int_\cE |\nabla_p Lu(s)|^2 \, \Pi_k + \int_\cE Lu(s) \, [L,T_\eta] u(s) \, \Pi_k 
  \end{aligned}
  \]
  where we have used 
  \begin{equation}
    \label{eq:T_eta_f2}
    \frac12 T_\eta(|f|^2) = f T_\eta f + \frac{\gamma}{\beta} |\nabla_p f|^2.
  \end{equation}
  In addition, denoting by $T_\eta^\dagger$ the adjoint of $T_\eta$ on the flat space $L^2(\mathcal{E})$,
  \[
  \int_\cE T_\eta\left( |Lu(s)|^2 \right) \, \Pi_k =  \int_\cE |Lu(s)|^2 \, T_\eta^\dagger\Pi_k.
  \]
  A simple computation shows that there exists $A_{k,\eta_*} > 0$ such that, for any $\eta \in [-\eta_*,\eta_*]$ (see the expression of $T_0^\dagger \Pi_k$ in the proof of~\cite[Lemma~25]{LMS13})
  \[
  T_\eta^\dagger \Pi_k\leq A_{k,\eta_*} \Pi_k.
  \]
  The inequality then follows by taking into account the extra factor $\rme^{\alpha s}$ and integrating in time. 
\end{proof}

In the remainder of the proof, the constants may vary from line to line. To use Lemma~\ref{lem:poly_decay}, we will need the following commutators:
\[
\begin{aligned}[] 
[\partial_{p_i},T_\eta] & = \sum_{j=1}^d \left(M^{-1}\right)_{ij} \left(\partial_{q_j}-\gamma \partial_{p_i}\right) \\
[\partial_{q_i},T_\eta] & = [\partial_{q_i}, (-\nabla V + \eta F)\cdot \nabla_p] = \left(-\nabla (\partial_{q_i}V) + \eta \partial_{q_i} F \right)\cdot \nabla_p. 
\end{aligned}
\]
We then start by applying Lemma~\ref{lem:poly_decay} with $L=\mathrm{Id}$ and $\alpha < 2 \lambda_n$, so that the right-hand side of inequality~\eqref{eq:exp_decay_bound} is uniformly bounded in time in view of the estimate~\eqref{eq:weighted_L2_decay}. This shows that 
\begin{equation}
\label{eq:bound_integrated_dp}
\int_0^s \rme^{\alpha r} \left( \int_\cE \left|\nabla_p u(r)\right|^2 \, \Pi_l \right)dr \leq K \| f\|_{L^\infty(L^\infty_{\Li_n})}^2
\end{equation}
is uniformly bounded for $s \geq 0$. We next consider $L_i = \partial_{q_i} - R \partial_{p_i}$ for some real parameter $R > 0$, which mixes derivatives in~$q$ and~$p$ in order to retrieve some dissipation in the~$q$ direction (this was one of the main ideas of hypocoercivity, as mentioned in the introduction of~\cite{Villani}). Note that
\[
[L_i,T_\eta] = \left(-\nabla (\partial_{q_i}V) + \eta \partial_{q_i} F \right)\cdot \nabla_p - R \sum_{j=1}^d \left(M^{-1}\right)_{ij} \left(\partial_{q_j}-\gamma \partial_{p_i}\right),
\]
so that, using a Cauchy-Schwarz inequality and the boundedness of $\nabla (\partial_{q_i} V) + \eta \partial_{q_i} F$, 
\[
\begin{aligned}
\sum_{i=1}^d L_i u(s) \, [L_i ,T_\eta]u(s) & = -R \left|M^{-1/2} (\nabla_q - R \nabla_p)u(s)\right|^2 + R (\gamma-R)\Big( (\nabla_q - R \nabla_p)u(s) \Big)^T M^{-1} \nabla_p u(s) \\
& \ \ \ + \sum_{i=1}^d L_i u(s) \, \left(-\nabla (\partial_{q_i}V) + \eta \partial_{q_i} F \right)^T \nabla_p u(s) \\
& \leq - \frac{R}{2\sqrt{m}} \left|(\nabla_q - R \nabla_p)u(s)\right|^2 + b_R |\nabla_p u(s)|^2
\end{aligned}
\]
for some constant $b_R > 0$ and where $M \leq m \mathrm{Id}$ (recall~\eqref{eq:bound_mass_matrix}). In addition, there exists an integer $n' \geq n$ such that $\nabla_q f, \nabla_p f \in L^\infty(L^\infty_{\mathcal{K}_{n'}})$. Upon choosing $l > 2n'+d/2$,
\[
\int_\mathcal{E} \left( |\nabla_p f|^2 + |\nabla_q f|^2 \right) \Pi_{l} \leq C \, \sup_{\substack{ r\in \mathbb{N}^{2d} \\ |r| = 1}} \left\| \partial^r f \right\|^2_{L^\infty(L^\infty_{\Li_n'})}.
\]
Lemma~\ref{lem:poly_decay} together with \eqref{eq:bound_integrated_dp} then shows that it is possible to choose $R > 0$ sufficiently large so that, upon increasing the value of~$K$ and choosing $l > 2n'+d/2$ (in view of the integral on the initial condition $u(0;t,q,p) = f(t,q,p) - \overline{f}_\eta(t)$), 
\[
\int_\cE \left|(\nabla_q-R\nabla_p) u(s)\right|^2 \, \Pi_l \leq K \, \rme^{-\alpha s} \, \sup_{\substack{ r\in \mathbb{N}^{2d} \\ |r| \leq 1}} \left\| \partial^r f \right\|_{L^\infty(L^\infty_{\Li_n'})}^2.
\]
Therefore, by integration (and upon changing $K$ again),
\begin{equation}
\label{eq:bound_integrated_dq-dp}
\int_0^s \rme^{\alpha r} \left( \int_\cE \left|(\nabla_q-R\nabla_p) u(r)\right|^2 \, \Pi_l\right)dr \leq K
\, \sup_{\substack{ r\in \mathbb{N}^{2d} \\ |r| \leq 1}} \left\| \partial^r f \right\|_{L^\infty(L^\infty_{\Li_n'})}^2
\end{equation}
uniformly in~$s \geq 0$. The combination of~\eqref{eq:bound_integrated_dp} and~\eqref{eq:bound_integrated_dq-dp} finally gives a control on $\nabla_q u$ as
\[
\int_0^s \rme^{\alpha r} \left( \int_\cE \left|\nabla_q u(r)\right|^2 \, \Pi_l\right)dr \leq K \, \sup_{\substack{ r\in \mathbb{N}^{2d} \\ |r| \leq 1}} \left\| \partial^r f \right\|_{L^\infty(L^\infty_{\Li_n'})}^2.
\]
We can now use Lemma~\ref{lem:poly_decay} with $L = \partial_{p_i}$ or $\partial_{q_i}$ (with a Cauchy-Schwarz inequality to treat the last term on the right-hand side of~\eqref{eq:exp_decay_bound}, namely the one involving the commutator) to conclude that
\[
\int_\cE \Big( |\nabla_p u(s)|^2+|\nabla_p u(s)|^2 \Big) \Pi_l \leq C \rme^{-\alpha s} \, \sup_{\substack{ r\in \mathbb{N}^{2d} \\ |r| \leq 1}} \left\| \partial^r f \right\|_{L^\infty(L^\infty_{\Li_n'})}^2.
\]
This is the extension of~\eqref{eq:weighted_L2_decay} taking into account first order derivatives. A similar bound for derivatives of arbitrary order is obtained by induction, as carefully presented in~\cite{Kopec13}. This allows to show that there exists $l,m \in \mathbb{N}$ and $\alpha > 0$ such that $\partial^{k_1} ( \partial^{k_2} u(s) \Pi_l) \in L^2(\cE)$ for any multi-indices~$k_1,k_2$ with $|k_1|+|k_2| \leq k_{\rm max}$, with a $L^2$ norm exponentially decreasing as a function of~$s$: 
\[
\left\| \partial^{k_1} ( \partial^{k_2} u(s) \Pi_l) \right\|_{L^2(\mathcal{E})} \leq C \, \rme^{-\alpha s} \, \sup_{\substack{ r\in \mathbb{N}^{2d} \\ |r| \leq k_{\rm max}}} \left\| \partial^r f \right\|_{L^\infty(L^\infty_{\Li_m})}.
\]
Sobolev embeddings then imply pointwise estimates of the form
\[
\| \partial^{k_2} u(s) \Pi_l \|_{L^\infty} \leq C \, \rme^{-\alpha s} \, \sup_{\substack{ r\in \mathbb{N}^{2d} \\ |r| \leq k_{\rm max}}} \left\| \partial^r f \right\|_{L^\infty(L^\infty_{\Li_m})},
\]
from which~\eqref{eq:pointwise_Poisson} is obtained upon integrating in~$s$ (see~\eqref{eq:def_Phi_eta_semigroup}).

\subsection{Proof of Theorem~\ref{thm:diffusion}}
\label{sec:proof_diffusion}

We introduce a given test direction $\xi \in \RR^d$ to project the evolution onto, and consider the solutions $\Phi_\eta = (\Phi_{\eta,1},\dots,\Phi_{\eta,d})$ of the following Poisson equations:
\[
(\partial_t + \cA_0 + \eta \cA_1) \Phi_{\eta,i} = \left( M^{-1}p \right)_i - \mathcal{V}_{\eta,i} = \sum_{j=1}^d \left(M^{-1}\right)_{i,j} p_j - \mathcal{V}_{\eta,i}, 
\qquad 
\qquad \int_\cE \Phi_{\eta,i} \, \psi_\eta = 0,
\]
where $\mathcal{V}_{\eta,i}$ denotes the $i$th component of the vector $\mathcal{V}_{\eta}$. By It\^o calculus (note that, in view of Proposition~\ref{prop:Poisson}, the derivatives of $\Phi_\eta$ are well defined elements of $L^\infty(L^\infty_{\Lin})$ for some integer $n$ sufficiently large, hence $\cA_0 \Phi_\eta$, etc, make sense in this functional space),
\[
d\left(\xi^T \Phi_\eta\right)([t],q^\eta_t,p^\eta_t) = \left[(\partial_t + \cA_0 + \eta \cA_1)\left(\xi^T\Phi_\eta\right)\right]([t],q^\eta_t,p^\eta_t) + \sqrt{\frac{2\gamma}{\beta}} \nabla_p \left(\xi^T\Phi_\eta\right)([t],q^\eta_t,p^\eta_t) \cdot dW_t.
\]
By definition of~$\Phi_\eta$,
\[
\begin{aligned}
\xi^T \Big( Q_{t}^{\eta,\eps} - Q_0^{\eta,\eps} \Big) & = \eps \int_0^{t/\eps^2} \xi^T\left(M^{-1}p^\eta_s - \mathcal{V}_{\eta} \right)ds \\
& = \eps \xi^T \left( \Phi_\eta\left(\left[\frac{t}{\eps^2}\right],q^\eta_{t/\eps^2},p^\eta_{t/\eps^2}\right)- \Phi_\eta\left(0,q^\eta_{0},p^\eta_{0}\right) \right) - \eps \mathscr{M}^{\eta,\xi}_{t/\eps^2},
\end{aligned}
\]
where we have introduced the martingale
\[
\mathscr{M}^{\eta,\xi}_s = \sqrt{\frac{2\gamma}{\beta}} \int_0^s \nabla_p \left(\xi^T \Phi_\eta\right)([\theta],q^\eta_\theta,p^\eta_\theta) \cdot dW_\theta.
\]
The quadratic variation of $\mathscr{M}^{\eta,\xi}_s$ is
\[
\left\langle \mathscr{M}^{\eta,\xi} \right\rangle_s = \frac{2\gamma}{\beta} \int_0^s \left|\nabla_p \left(\xi^T \Phi_\eta\right)([\theta],q^\eta_\theta,p^\eta_\theta)\right|^2 \, d\theta.
\]
The proof of the weak convergence to a Brownian motion is very standard: we prove the convergence of the finite dimensional laws using the martingale central limit theorem~\cite{KomorowskiLandimOlla2012, EthKur86}, as well as the tightness of the stochastic process, and conclude with~\cite[Theorem~7.1]{Billingsley99}. Alternatively, we could also follow the more quantitative approach from~\cite{HP04} (based on the use of an additional Poisson equation to get better control between $\xi^T Q^{\eta, \epsilon}_t$ and the limiting process, projected onto $\xi$) in view of the good resolvent estimates provided by Proposition~\ref{prop:Poisson}.

\paragraph{Tightness.}
Let us first prove the tightness by using Prohorov's theorem. It is sufficient to prove that, for any $\alpha > 0$ and any $\tau > 0$,
\begin{equation}
\label{eq:tightness}
\lim_{\delta \to 0} \limsup_{\eps \to 0} \mathbb{P}\left( \sup_{\substack{ |t-s|<\delta \\ 0 \leq s < t \leq \tau}} \left| \xi^T \Big( Q_{t}^{\eta,\eps} - Q_s^{\eta,\eps} \Big)\right| \geq \alpha \right) = 0. 
\end{equation}
Note that, using for instance the $L^\infty(L^\infty_{\Li_2})$ bound, \textit{i.e.}~\eqref{eq:Linf_estimate_Phi} with $n=2$,
\[
\begin{aligned}
& \mathbb{P}\left( \sup_{\substack{ |t-s|<\delta \\ 0 \leq s < t \leq \tau}} \eps\left| \xi^T \left( \Phi_\eta\left(\left[\frac{t}{\eps^2}\right],q^\eta_{t/\eps^2},p^\eta_{t/\eps^2}\right)- \Phi_\eta\left(\left[\frac{s}{\eps^2}\right],q^\eta_{s/\eps^2},p^\eta_{s/\eps^2}\right) \right) \right| \geq \alpha \right) \\
& \qquad \leq \mathbb{P}\left( \sup_{0 \leq t \leq T} \eps\left| \xi^T \Phi_\eta\left(\left[\frac{t}{\eps^2}\right],q^\eta_{t/\eps^2},p^\eta_{t/\eps^2}\right) \right| \geq \frac{\alpha}{2} \right) \\
& \qquad \leq \mathbb{P}\left( \sup_{0 \leq t \leq \tau} 1 + |p^\eta_{t/\eps^2}|^2 \geq \frac{\alpha}{2\eps |\xi|\| \Phi_\eta \|_{L^\infty(L^\infty_{\Li_2})}} \right) \\
& \qquad \leq \sup_{0 \leq t \leq \tau} \int_{|p|>\alpha/2\eps |\xi|\| \Phi_\eta \|_{L^\infty(L^\infty_{\Li_2})} } \Big( 1+|p|^2 \Big) \psi_\eta(t,q,p) \, dq \, dp \xrightarrow[\eps \to 0]{} 0,
\end{aligned}
\]
uniformly in~$\delta > 0$ (using the finiteness of the moments~\eqref{eq:finite_moments} and the smoothness of $\psi_\eta$). On the other hand, by Doob's inequality,
\[
\begin{aligned}
\mathbb{P}\left( \sup_{\substack{ |t-s|<\delta \\ 0 \leq s < t \leq \tau}} \eps\left| \mathscr{M}^{\eta,\xi}_{t/\eps^2}-\mathscr{M}^{\eta,\xi}_{s/\eps^2} \right| \geq \alpha \right) & \leq \frac\eps\alpha \sqrt{\mathbb{E}\left(\left| \mathscr{M}^{\eta,\xi}_{t/\eps^2}-\mathscr{M}^{\eta,\xi}_{s/\eps^2} \right|^2\right)} \\
& =  \frac{\eps}{\alpha} \sqrt{\frac{2\gamma}{\beta}\int_{s/\eps^2}^{t/\eps^2} \mathbb{E}\left(\left|\nabla_p \left(\xi^T \Phi_\eta\right)\left(\theta,q^\eta_\theta,p^\eta_\theta\right)\right|^2\right) d\theta} \\
& \leq \frac{\eps}{\alpha}  \, \sqrt{2T\, \xi^T \mathscr{D}_\eta \xi\left(\left\lceil \frac{t}{T \eps^2 }\right\rceil - \left\lfloor \frac{s}{T \eps^2 }\right\rfloor\right) } \\
& \leq \frac{1}{\alpha} \, \sqrt{ 2\left( t-s + 2T\eps^2 \right)\xi^T \mathscr{D}_\eta \xi },
\end{aligned}
\]
where the final bound relies on $\lceil x \rceil = [x] + 1 \leq x + 1$ and $\lfloor x \rfloor = [x] - 1 \geq x - 1$. To go from the second to the third line, we have used the trivial bound
\[
\int_{s/\eps^2}^{t/\eps^2} \mathbb{E}\left(\left|\nabla_p \left(\xi^T \Phi_\eta\right)([\theta],q^\eta_\theta,p^\eta_\theta)\right|^2\right) d\theta \leq \int_{\lfloor s/T\eps^2 \rfloor T}^{\lceil t/T\eps^2 \rceil T} \mathbb{E}\left(\left|\nabla_p \left(\xi^T \Phi_\eta\right)([\theta],q^\eta_\theta,p^\eta_\theta)\right|^2\right) d\theta,
\]
and
\[
\begin{aligned}
\int_{nT}^{(n+1)T} \mathbb{E}\left(\left|\nabla_p \left(\xi^T \Phi_\eta\right)([\theta],q^\eta_\theta,p^\eta_\theta)\right|^2\right) d\theta & = T \int_{T\mathbb{T} \times \M \times \RR^d} \left|\nabla_p \left(\xi^T \Phi_\eta\right)([t],q,p)\right|^2 \psi_\eta(t,q,p) \, dt \, dq \, dp \\
& = \frac{\beta T}{\gamma} \xi^T \mathscr{D}_\eta \xi.
\end{aligned}
\]
This owes to the fact that the law of $(q^\eta_\theta,p^\eta_\theta)$ is $\psi_\eta([\theta],q,p)$ since the dynamics is started at time~0 from $\psi_\eta(0,q,p) \, dq\,dp$. Combining the estimates we have shown, \eqref{eq:tightness} easily follows.

\paragraph{Convergence of finite dimensional distributions.}
We first rewrite $\xi^T \Big( Q_{t}^{\eta,\eps} - Q_0^{\eta,\eps} \Big)$ as the sum of discrete, stationary martingale increments, and a remainder term. In order to have stationary increments, we consider evolutions on integer multiples of the period, and introduce, for $k \in \mathbb{N}$, 
\[
Z_k = \sqrt{\frac{2\gamma}{\beta}} \int_{kT}^{(k+1)T} \nabla_p \left(\xi^T \Phi_\eta\right)(\theta,q^\eta_\theta,p^\eta_\theta) \cdot dW_\theta.
\]
Define $n_\eps(t) = \lfloor t/(T\eps^2) \rfloor$. Then, 
\[
\xi^T \Big( Q_{t}^{\eta,\eps} - Q_0^{\eta,\eps} \Big) = R_t^\eps + \frac{\sqrt{n_\eps(t)}}{\eps} \, \mathfrak{M}_t^\eps, \qquad \mathfrak{M}_t^\eps = \frac{1}{\sqrt{n_\eps(t)}} \sum_{k=1}^{n_\eps(t)} Z_k,
\]
with $\eps \sqrt{n_\eps(t)} \to \sqrt{t/T}$ and
\[
R_t^\eps = \eps \xi^T \left( \Phi_\eta\left(\frac{t}{\eps^2},q^\eta_{t/\eps^2},p^\eta_{t/\eps^2}\right)- \Phi_\eta\left(0,q^\eta_{0},p^\eta_{0}\right) \right) - \eps \left( \mathscr{M}^{\eta,\xi}_{t/\eps^2} -  \mathscr{M}^{\eta,\xi}_{n_\eps(t)T} \right).
\]
It is easy to check that $\mathbb{E}|R_t^\eps|^2 \to 0$ as $\eps \to 0$ with computations similar to the ones used in the proof of the tightness. Hence, the convergence of finite dimensional distributions is dictated by the sum of martingale increments~$\mathfrak{M}_t^\eps$. To obtain the asymptotic behavior of this random variable, we resort for instance to~\cite[Theorem~35.12]{Billingsley95} (see also~\cite[Theorem~3.2]{Helland82}). Denoting by $ \mathcal{F}_{k-1}$ the filtration of events up to time~$kT$,
\[
\mathbb{E}\left[Z_k^2 \, \Big| \, \mathcal{F}_{k-1} \right] = \frac{2\gamma}{\beta} \int_{kT}^{(k+1)T} \left|\nabla_p \left(\xi^T \Phi_\eta\right)([t],q^\eta_s,p^\eta_s)\right|^2 \,ds,
\]
the first condition of~\cite[Theorem~35.12]{Billingsley95} follows from the Law of Large Numbers~\eqref{eq:LLN}:
\[
\frac{1}{n_\eps(t)} \sum_{k=1}^{n_\eps(t)} \mathbb{E}\left[Z_k^2 \, \Big| \, \mathcal{F}_{k-1} \right] = \frac{2\gamma}{\beta n_\eps(t)} \int_{0}^{n_\eps(t) T} \left|\nabla_p \left(\xi^T \Phi_\eta\right)([t],q^\eta_s,p^\eta_s)\right|^2 \,ds \xrightarrow[t \to +\infty]{} 2T \, \xi^T \mathscr{D}_\eta \xi \quad \mathrm{a.s.}
\]
On the other hand, for any $\alpha > 0$,
\[
\mathbb{E}\left(Z^2_1 \mathbf{1}_{|Z_1| > \sqrt{n} \alpha}\right) \xrightarrow[n \to +\infty]{} 0
\]
by dominated convergence, which implies the second condition of~\cite[Theorem~35.12]{Billingsley95}. In conclusion, $\xi^T \Big( Q_{t}^{\eta,\eps} - Q_0^{\eta,\eps} \Big)$ converges in law to a standard Gaussian distribution, with variance $2t \, \xi^T \mathscr{D}_\eta \xi$. The extension to a general finite dimensional distributions is done using an iterative procedure, as carefully documented in~\cite[Section~3.4]{BLP11} for instance.

\paragraph{Positive definiteness of the covariance matrix.} It is clear that $\mathscr{D}_\eta$ is symmetric positive. To prove that it is definite, assume that there exists~$\xi \neq 0$ such that $\xi^T \mathscr{D}_\eta \xi = 0$, in which case
\[
\int_\mathcal{E} \left|\nabla_p \left( \xi^T \Phi_\eta \right)\right|^2 \psi_\eta = 0.
\]
This shows that $\nabla_p\left( \xi^T \Phi_\eta \right) = 0$ almost everywhere since $\psi_\eta(t,q,p) > 0$ by Proposition~\ref{prop:inv_meas_cv}, hence $\xi^T \Phi_\eta$ does not depend on~$p$. But then, the equality 
\[
(\partial_t + \cA_0 + \eta \cA_1)\left( \xi^T \Phi_\eta \right) = (\partial_t + p^T M^{-1} \nabla_q)\left( \xi^T \Phi_\eta \right) = \xi^T(M^{-1}p - \overline{v}_\eta(t))
\]
requires $\nabla_q\left( \xi^T \Phi_\eta \right) = \xi$, 
which cannot be satisfied for a smooth, periodic function of the~$q$ variable. The contradiction shows that~$\mathscr{D}_\eta$ is definite.

\subsection{Proof of Proposition~\ref{prop:small_eta_ppties_D}}

Since the operator $T_\eta$ defined in~\eqref{eq:def_T_eta} satisfies~\eqref{eq:T_eta_f2}, and using, by definition of the invariant measure,
\[
\int_\cE T_\eta\Big(|f|^2\Big)\psi_\eta = 0,
\]
a simple computation shows that, for fixed direction~$\xi \in \RR^d$,
\begin{align}
\xi^T \mathscr{D}_\eta \sL & = - \int_\cE \left(\sL^T \Phi_\eta\right) \sL^T\left(M^{-1}p - \mathcal{V}_\eta\right) \psi_\eta(t,q,p) \, dt \, dq \, dp \label{eq:equation_to_expand} \\
& = \sL^T \left( \int_0^{+\infty} \mathbb{E}_\eta\Big( \left(M^{-1}p^\eta_s - \mathcal{V}_\eta\right) \otimes \left(M^{-1}p_0^\eta - \mathcal{V}_\eta\right) \Big)ds \right)\sL, \nonumber
\end{align}
where the expectation $\mathbb{E}_\eta$ is with respect to initial conditions $(t_0,q_0,p_0)$ distributed according to~$\psi_\eta$, and for all realizations of the extended dynamics~\eqref{eq:Langevin_extended}. The small~$\eta$ behavior of~$\mathscr{D}_\eta$ is deduced from~\eqref{eq:equation_to_expand}, by expanding the various terms depending on~$\eta$. The expansion of the invariant measure has already been established in~\eqref{eq:formal_expansion_varrho}. This also gives an expansion for the average velocity: 
\begin{equation}
\label{eq:expansion_V_eta}
\mathcal{V}_\eta = \eta \overline{\mathscr{V}} + \eta^2 \widetilde{\mathcal{V}}_\eta^2, 
\end{equation}
where $\overline{\mathscr{V}}$ is defined in~\eqref{eq:time_avg_response} and $\widetilde{\mathcal{V}}_\eta^2$ collects the higher order terms and is uniformly bounded for $|\eta|$ sufficiently small. The expansion of the solution of the Poisson equation is given by the following lemma.

\begin{lemma}
\label{lem:expansion_Phi_eta}
The solution $\Phi_\eta$ of the Poisson equation~\eqref{eq:Poisson_eq_velocity} can be expanded as
\[
\Phi_\eta = \Phi_0 - \eta \widetilde{\Phi}^1 + \eta^2 \widetilde{\Phi}^2_\eta,
\]
where $\Phi_0$ is defined in~\eqref{eq:Phi_0} and $\widetilde{\Phi}^1$ satisfies the Poisson equation
\begin{equation}
\label{eq:def_first_orders_Phi_eta}
\qquad (\partial_t + \cA_0) \widetilde{\Phi}^1 = \cA_1 \Phi_0 - \int_\cE \cA_1 \Phi_0 \, \mu, \qquad \int_\cE \widetilde{\Phi}^1(t,q,p) \, \mu(q,p) \, dt \, dq \, dp = 0,
\end{equation}
while there exists some $m \geq 1$ such that the remainder $\widetilde{\Phi}^2_\eta$ is uniformly bounded in the $L^\infty(L^\infty_{\mathcal{K}_m})$ norm for $|\eta|$ sufficiently small.
\end{lemma}

Gathering all these expansions, we see that
\[
\sL^T \mathscr{D}_\eta \sL = \sL^T \mathscr{D}_0 \sL + \eta \sL^T \mathpzc{D}_1 \sL + \eta^2 \widetilde{\mathpzc{D}}_{\eta,\sL},
\]
with 
\[
\begin{aligned}
\sL^T \mathpzc{D}_1 \sL & = \int_\cE \left(\sL^T \widetilde{\Phi}^1(t,q,p)\right) \Big( \sL^T M^{-1}p \Big) \mu(q,p) \, dt \, dq \, dp \\
& \ \ - \int_\cE \left(\sL^T \Phi_0(q,p) \right) \left( \sL^T M^{-1}p \, \varrho_1(t,q,p)  - \overline{\mathscr{V}}\right) \mu(q,p) \, dt \, dq \, dp, 
\end{aligned}
\]
and where $\widetilde{\mathpzc{D}}_{\eta,\sL}$ is uniformly bounded for $|\eta| \leq r$ and $|\sL| \leq 1$. 

\bigskip

We now turn to the specific case when the condition~\eqref{eq:null_time_avg_F} holds. Note that there is only one function depending on time in each term on the right-hand side of the definition of $\sL^T \mathpzc{D}_1 \sL$. The idea is to show that the time average of these functions is~0, which will then give the claimed result. We define the time-average of a function on~$\mathcal{E}$ as
\[
\widehat{f}(q,p) = \frac1T \int_{T \mathbb{T}} f(t,q,p) \, dt.
\]
When \eqref{eq:null_time_avg_F} holds, an integration of~\eqref{eq:eq_varrho_1} with respect to the time variable gives
\[
\cA_0^* \widehat{\varrho}_1 = 0,
\]
which shows that $\widehat{\varrho}_1$ is constant since $\mathrm{Ker}(\cA_0^*) = \mathrm{Span}(\mathbf{1})$, and finally that this function is~0 by the normalization condition~\eqref{eq:conditions_varrho}. This also implies that $\overline{\mathscr{V}} = 0$. Similarly, by integrating in time the Poisson equation~\eqref{eq:def_first_orders_Phi_eta} satisfied by~$\widetilde{\Phi}^1$,
\[
\cA_0 \widehat{\widetilde{\Phi}^1} = 0,
\]
which also gives $\widehat{\widetilde{\Phi}^1} = 0$. Finally, it indeed holds $\sL^T \mathpzc{D}_1 \sL = 0$ when~\eqref{eq:null_time_avg_F} holds.

\bigskip

\begin{proof}[Proof of Lemma~\ref{lem:expansion_Phi_eta}]
The function $\widetilde{\Phi}^1$ introduced in~\eqref{eq:def_first_orders_Phi_eta} is indeed well defined in view of Lemmas~\ref{lem:pt_A0_bounded}. A simple computation shows that
\begin{equation}
\label{eq:formal_comput_order_2}
T_\eta \left(\Phi_\eta - \Phi_0 + \eta \widetilde{\Phi}^1 \right) = \eta^2 \left( -\widetilde{\mathcal{V}}_\eta^2 + \cA_1 \widetilde{\Phi}^1 \right).
\end{equation}
Indeed, the linear term disappears since (see~\eqref{eq:eq_varrho_1}, \eqref{eq:Phi_0} and the discussion after this equation)
\[
\begin{aligned}
\overline{\mathscr{V}} & = \int_\cE M^{-1}p \, \varrho_1(t,q,p) \, \mu(q,p) \, dt \, dq \, dp = -\int_\cE \left[(\partial_t + \cA_0)^{-1} \Big( M^{-1}p \Big)\right] \cA_1^*\mathbf{1} \, \mu(q,p) \, dt \, dq \, dp \\
& \ \ = -\int_\cE \Phi_0(q,p) \cA_1^*\mathbf{1} \, \mu(q,p) \, dt \, dq \, dp = -\int_\cE (\cA_1 \Phi_0)(t,q,p) \, \mu(q,p) \, dt \, dq \, dp,
\end{aligned}
\]
which is the constant term on the right-hand side of the Poisson equation defining~$\widetilde{\Phi}^1$ in~\eqref{eq:def_first_orders_Phi_eta}. We then conclude from~\eqref{eq:formal_comput_order_2} and the estimates provided for instance by Proposition~\ref{prop:Poisson} in the case $\eta = 0$. These estimates indeed show that $\Phi_0$ (and hence $\mathcal{A}_1 \Phi_0$) is smooth with derivatives growing at most polynomially. Next, from~\eqref{eq:def_first_orders_Phi_eta}, $\widetilde{\Phi}^1$ is smooth and this function and its derivatives grow at most polynomially.
\end{proof}

\subsection*{Acknowledgements}

This work was initiated while GP was visiting the INRIA team MICMAC (now MATHERIALS) at CERMICS. The hospitality and financial support from INRIA are greatly acknowledged. RJ's research is supported by the EPSRC through grant EP/J009636/1. GP's research is partially supported by the EPSRC through grants EP/J009636/1 and EP/H034587/1. GS's research is partially supported by the European Research Council under the European Union's Seventh Framework Programme (FP/2007-2013) / ERC Grant Agreement number 614492. The authors benefited from discussions with Stefano Olla and Stephan De Bi\`evre.



\end{document}